\documentclass[11pt]{amsart}
\usepackage[
  left=4cm,
  right=4cm,
  top=3cm,
  bottom=3cm
]{geometry}

\usepackage{hyperref}
\hypersetup{linktocpage = true, colorlinks = true, linkcolor = blue, citecolor= red, urlcolor = green}

\usepackage[backend=biber, style=numeric, sorting=none, doi=true, url=true, eprint=true]{biblatex}
\usepackage{tikz}
\usetikzlibrary{positioning, quotes, arrows, decorations.markings, decorations.pathmorphing, decorations.pathreplacing, shapes, patterns, calc}
\usepackage{url}

\usepackage{graphicx, color, bm}
\usepackage{array,tikz,tikz-cd,float}

\usepackage{amsmath,amsthm}
\usepackage{amssymb, amsfonts, verbatim, subfigure}

\usepackage[most]{tcolorbox}
\usepackage[bbgreekl]{mathbbol}
\usepackage{mathpazo}
\usepackage[mathscr,mathcal]{euscript}

\tikzset{middlearrow/.style={
		decoration={markings,
			mark= at position 0.5 with {\arrow{#1}} ,
		},
		postaction={decorate}
	}
}

\newcommand{\id}{\mathbb{1}}

\newcommand{\be}{\begin{equation}}
	\newcommand{\ee}{\end{equation}}
\newcommand{\beqn}{\begin{equation}}
	\newcommand{\eeqn}{\end{equation}}
\newcommand{\bp}{\begin{pmatrix}}
	\newcommand{\ep}{\end{pmatrix}}
\newcommand{\bsp}{\left(\begin{smallmatrix}}
	\newcommand{\esp}{\end{smallmatrix}\right)}

\newcommand{\R}{{\mathbb R}}

\newcommand{\C}{{\mathbb C}}

\newcommand{\tL}{{\tilde L}}
\newcommand{\ta}{{\tilde \alpha}}
\newcommand{\tg}{{\tilde \Gamma}}
\newcommand{\tH}{{\tilde H}}

\newcommand{\dr}{{\mathrm{d}}}

\newcommand{\CA}{{\mathcal A}}
\newcommand{\CB}{{\mathcal B}}

\newcommand{\cF}{{\mathcal F}}

\newcommand{\CJ}{{\mathcal J}}
\newcommand{\CK}{{\mathcal K}}

\newcommand{\CV}{{\mathcal V}}

\newcommand{\fg}{\mathfrak{g}}
\newcommand{\fn}{\mathfrak{n}}

\newcommand{\fh}{\mathfrak{h}}
\newcommand{\fG}{\mathfrak{G}}

\newcommand{\sfx}{\mathsf{x}}
\newcommand{\sfy}{\mathsf{y}}

\numberwithin{equation}{section}
\numberwithin{figure}{section}
\numberwithin{table}{section}

\renewcommand{\d}{\mathrm{d}}

\newcommand{\cA}{\mathcal{A}}

\newtheoremstyle{thm}
{7pt}
{7pt}
{\itshape}
{}
{\bf}
{.}
{5pt}
{\thmnumber{#2 }\thmname{#1}\thmnote{ (#3)}}

\newtheoremstyle{def}
{7pt}
{10pt}
{\itshape}
{}
{\bf}
{.}
{5pt}
{\thmnumber{#2} \thmname{#1}\thmnote{ (#3)}}

\newtheoremstyle{rem}
{4pt}
{10pt}
{}
{}
{\itshape}
{:}
{3pt}
{}

\newtheoremstyle{texttheorem}
{8pt}
{8pt}
{\itshape}
{}
{\bf}
{. \hspace{5pt}}
{3pt}
{}

\theoremstyle{thm}

\newtheorem*{theorem*}{Theorem}
\newtheorem*{lemma*}{Lemma}
\newtheorem*{corollary*}{Corollary}
\newtheorem*{proposition*}{Proposition}
\newtheorem*{definition*}{Definition}

\newtheorem{theorem}{Theorem}[subsection]
\newtheorem{thm-def}{Theorem/Definition}[theorem]
\newtheorem{prop}[theorem]{Proposition}

\newtheorem*{question*}{Question}
\newtheorem{lemma}[theorem]{Lemma}

\numberwithin{equation}{subsection}

\theoremstyle{definition}
\newtheorem{dfn}[theorem]{Definition}

\theoremstyle{rem}

\newtheorem{remark}{Remark}

\usepackage{stmaryrd}
\date{}

\title{Infinite Dimensional Topological-Holomorphic Symmetry in Three-Dimensions}

\author{Hank Chen}
\address{Beijing Institute of Mathematical Sciences and Applications, Beijing 101408, China}
\email{chunhaochen@bimsa.cn} 

\author{Joaquin Liniado}
\address{Instituto de Física La Plata, UNLP and CONICET, C.C. 67, 1900 La Plata, Argentina}
\address{School of Mathematics and Maxwell Institute for Mathematical Sciences, University of Edinburgh, EH9 3FD, United Kingdom}
\email{jliniado@ed.ac.uk}

\begin{document}
	
\begin{abstract}

We introduce a three-dimensional quantum field theory with an infinite-dimensional symmetry, realized explicitly through a centrally extended affine graded Lie algebra. This symmetry is a direct three-dimensional generalization of the chiral symmetry in the Wess-Zumino-Witten model. Upon performing radial quantization, we construct the Fock space of the theory and, via a three-dimensional analogue of the state–operator correspondence, we demonstrate that the algebra of local operators is endowed with the structure of a raviolo vertex algebra. Accordingly, this setup provides a framework for extending the methods of two-dimensional conformal field theory to three dimensions, and we expect it to lay the groundwork for exact methods in three-dimensional quantum field theory.
	
\end{abstract}

\maketitle

\tableofcontents

\section{Introduction}

Two-dimensional conformal field theories have become a foundational framework for both physics and mathematics. The modern formulation of conformal invariance in two dimensions began with the seminal  work of Belavin, Polyakov, and Zamolodchikov \cite{Belavin:1984vu}. Building on the representation theory of the Virasoro algebra, they introduced the concept of an algebra of local operators and used it to construct exactly solvable conformal theories, known as minimal models. This breakthrough sparked intense activity at the intersection of mathematical physics and statistical mechanics, as minimal models were soon identified with a variety of two-dimensional statistical systems at criticality. 

In view of these developments, many efforts were devoted to extending the powerful tools of two-dimensional conformal field theory to higher dimensions. These attempts, however, encountered fundamental obstacles. The most immediate is Liouville’s rigidity theorem, which implies that for $d\geq 3$, all local conformal transformations extend to global ones, and the global conformal group is finite-dimensional. This rules out the existence of local infinite-dimensional conformal symmetry algebras in higher dimensions.

Given these obstructions, a prevailing strategy has been to isolate 2d chiral subalgebras within higher-dimensional theories, where infinite-dimensional symmetries can still be realized. A prominent example is provided by $4$d \(\mathcal{N}=2\) superconformal field theories, where a protected sector of local operators was shown in \cite{Beem:2013sza} to organize into a two-dimensional vertex operator algebra, enabling the use of chiral algebra techniques to obtain exact results. 

The aim of this work is to revisit the pursuit of infinite-dimensional symmetry algebras in three-dimensional quantum field theories, but with a shift in perspective. Rather than looking for a subsector whose dynamics can be captured by a two-dimensional chiral algebra, the strategy is to generalize the very notion of a chiral algebra into an infinite-dimensional symmetry structure suited to describe three-dimensional dynamics. These algebras enjoy a generalized notion of chirality while remaining infinite-dimensional, and thus preserve many of the powerful features typically associated with chiral algebras.

The three-dimensional theories that we will consider are partially topological and partially holomorphic. First introduced in the context of twisted 3d $\mathcal{N}=2$ supersymmetric field theories \cite{Aganagic:2017tvx}, they have since generated considerable interest and have been studied from a variety of mathematical and physical viewpoints. These include, in particular, Poisson vertex algebras \cite{Oh:2019mcg, Zeng:2021zef, Khan:2025rah}, factorization algebras \cite{Wang:2024tjf}, and twistorial field theories \cite{Garner:2023izn}, alongside more general investigations of their structure and properties \cite{Beem:2018fng,Costello:2020ndc, Gwilliam:2021zkv, Garner:2023wrc, Garner:2023zko, Dimofte:2025oqf, Budzik:2025zvu}.
We will focus on a specific example introduced by the authors in \cite{Chen:2024axr}, obtained via a localization procedure applied to holomorphic 2-Chern-Simons theory. 

These theories are defined on three-dimensional manifolds equipped with a transverse holomorphic foliation (THF) \cite{Duchamp:1979, Rawnsley}, meaning they are locally modeled on 
$\C\times \R$. They may be thought of as the mildest possible departure from the purely holomorphic two-dimensional setting: the holomorphic direction is retained, while a single additional real direction is introduced, but only topologically.

In two dimensions, the infinite-dimensional nature of chiral symmetry is rooted in the Laurent expansion of chiral fields in a punctured neighbourhood of a local operator insertion. In the present setting, however, one must understand what replaces this local description. While the complex coordinate continues to carry holomorphic structure, the real direction, being topological, contributes only through a residual notion of relative ordering. The corresponding local model is thus no longer the punctured disk, but its topological-holomorphic analogue, the raviolo, obtained by gluing two disks along a punctured disk (see \S \ref{raviolo} for details).

The natural question is therefore how holomorphicity should be reformulated in this topological-holomorphic setting. In the two-dimensional case, holomorphic functions on the punctured disk are captured by the $\bar\partial$-cohomology:
\begin{equation}
H^{(0,0)}_{\bar\partial}(\mathbb C\setminus\{0\}) \cong \mathbb C[z,z^{-1}] \, .
\end{equation}
This is what underlies the Laurent expansion of chiral fields, and hence the infinite tower of modes underlying chiral symmetry. The problem is then to identify the corresponding cohomology in the topological-holomorphic setting.

In \cite{Garner:2023zqn}, Garner and Williams address precisely this problem by identifying the relevant analogue of holomorphic functions on the raviolo in terms of the cohomology of the differential
\begin{equation}
    \mathrm{d}' = \bar{\partial} + \mathrm{d}_\tau \,,
\end{equation}
where $(z,\bar z,\tau)$ denote local coordinates on $\C\times \R$. They show that the degree-zero and degree-one cohomology of $\mathrm d'$ each give rise to spaces of formal series. In this way, $\mathrm d'$-cohomology provides the appropriate notion of local series expansion in the topological-holomorphic setting, and hence the natural framework in which the modes of $\mathrm d'$-closed symmetry currents may be defined. This is the basic mechanism through which infinite-dimensional symmetry algebras can emerge in three dimensions.

The general picture described above is realized concretely in the three dimensional field theory introduced in \cite{Chen:2024axr}, which will be the main object of study in this paper. The symmetry currents of that theory satisfy, schematically, a differential constraint of the form 
$\d'J=0$, and it was already anticipated there that they should generate an infinite-dimensional symmetry algebra. Since this is precisely the condition for the symmetry currents to define 
$\d'$-cohomology classes, the cohomological picture developed in \cite{Garner:2023zqn} provides the natural setting in which their mode algebra can be made explicit.

Concretely, the condition 
$\d'J=0$ allows one to expand the symmetry currents into formal series whose coefficients define an infinite collection of conserved charges, in direct analogy with the mode expansion of chiral currents in two-dimensions. In this context, one can also perform radial quantisation, thereby allowing the computation of commutators between the modes of the symmetry currents. The commutation relations are then shown to correspond to a centrally extended affine graded Lie algebra, which can be understood as a three-dimensional counterpart of the Kac–Moody algebra.

The existence of this algebra makes it possible to construct the Fock space of the theory by consistently selecting the creation operators among the modes of the symmetry currents. In the language of representation theory, this Fock space corresponds to the vacuum module of the algebra. As in the two-dimensional case, one can show that the states of this vacuum module are in one-to-one correspondence with local operators of the three-dimensional theory --- the so-called \textit{state-operator correspondence}.
Altogether, these ingredients define a raviolo vertex algebra in the sense of \cite{Garner:2023zqn}, which can be regarded as a three-dimensional generalisation of an ordinary vertex algebra.

The theory studied in this paper should be viewed as a concrete and explicit realization of a three-dimensional topological-holomorphic quantum field theory whose symmetry currents are organized by a raviolo vertex algebra. In close analogy with two-dimensional conformal field theories, where an infinite-dimensional chiral symmetry provides a powerful organizing principle and, in favourable cases, leads to exact results, the presence of a raviolo vertex algebra suggests that similar mechanisms may also be available in three dimensions. While this perspective is still at an early stage, the present example provides a hands-on setting in which these ideas can be developed concretely.

\subsection*{Outline}

The structure of the paper is as follows.
In Section \ref{sec:3dtheory}, we introduce the three-dimensional theory constructed in \cite{Chen:2024axr}, and describe its symmetry structure. In Section \ref{raviolo}, we present the raviolo formalism developed in \cite{Garner:2023zqn} and adapt it to the three-dimensional setting relevant to our theory. In Section \ref{affineraviolo}, we perform the radial quantisation of the three-dimensional theory and compute the commutation relations between the modes of the symmetry currents. In Section \ref{sec:ceagla}, we show that these commutation relations define a centrally extended affine graded Lie algebra, which we construct in detail. Finally, in Section \ref{sec:vertexalgebra} we construct the raviolo vertex algebra corresponding to the three-dimensional theory. 

\section{3d Topological-Holomorphic Theory}

\label{sec:3dtheory}

In this section, we introduce the main features of the three-dimensional theory constructed in \cite{Chen:2024axr}. Our goal is to show that the raviolo formalism introduced in \cite{Garner:2023zqn} is not merely a convenient tool, but in fact the natural framework for understanding the infinite dimensional symmetry structure of this theory.


Let $\fG=(\fh\xrightarrow{\mu_1}\fg,\mu_2)$ denote a (real) Lie 2-algebra and let $\mathbb{G}=(\mathsf{H}\xrightarrow{\mathsf{t}}G,\rhd)$ the corresponding Lie 2-group. A brief review of Lie 2-algebras and Lie 2-groups is provided in Appendix~\ref{sec:appLie2alg}. We consider the three-dimensional manifold $M=\mathbb{C}\times \R$. The three-dimensional action is given by \cite{Chen:2024axr}
\begin{multline}
       S[g,\Theta] = \int_M \operatorname{vol}_M \Big(\langle \partial_z (\partial_{\bar z}gg^{-1}), g\rhd\Theta_\tau\rangle - \langle \partial_z (\partial_\tau gg^{-1}),g\rhd\Theta_{\bar z}\rangle \\
       +\langle \mu_1\big(\partial_z (g\rhd\Theta_{\bar z})\big),g\rhd\Theta_\tau\rangle\Big)\,,\label{3daction-coords} 
\end{multline}
where $g\in C^{\infty}(M,G)$ is a $G$-valued smooth field, $\Theta \in \Omega^{1}(M,\fh)$ an $\fh$-valued $1$-form, and $\langle\cdot,\cdot\rangle:\fg \times \fh \to \mathbb{C}$ a degree $+1$, non degenerate, invariant bilinear form.


The appearance of the 1-form field \(\Theta\), not as a gauge connection but rather as a matter field playing a role analogous to that of \(g\), is closely tied to the higher categorical nature of the action. This reflects the categorical ladder = dimensional ladder principle \cite{Crane:1994ty, Baez:1995xq}, which underlies the construction of the three-dimensional action \eqref{3daction-coords}. As we shall see, the presence of this higher-form field, together with the graded structure provided by the Lie 2-algebra, integrates seamlessly with the raviolo formalism developed in \cite{Garner:2023zqn}.

The action \eqref{3daction-coords} is invariant under both left- and right-acting \emph{semi-local}\footnote{By semi-local we mean that these transformations are not gauge symmetries, but neither are they fully global. They exhibit a restricted coordinate dependence, subject to specific constraints.} symmetries. A direct computation shows that it is invariant under the transformations
\begin{equation}
\label{ec:3dleftsym}
    g \mapsto g(1 + \tilde\alpha) \,, \quad \Theta \mapsto \Theta-\tilde\alpha \rhd \Theta + \tilde\Gamma \,,
\end{equation}
as long as the infinitesimal parameters \(\tilde\alpha \in C^{\infty}(M, \fg)\) and \(\tilde\Gamma \in \Omega^{1}(M, \fh)\) satisfy the constraints
\begin{equation}
\label{ec:3drightsymconst}
    \partial_{\bar z} \tilde\alpha + \mu_1(\tilde\Gamma_{\bar z}) = 0\,, \quad 
    \partial_{\tau} \tilde\alpha + \mu_1(\tilde\Gamma_{\tau}) = 0\,, \quad 
    \partial_\tau \tilde\Gamma_{\bar z} - \partial_{\bar z} \tilde\Gamma_\tau = 0\,.
\end{equation}
These define the right-acting symmetries of the model. The left-acting symmetries take the form
\begin{equation}
\label{ec:3drightsym}
    g \mapsto (1 - \alpha)g \,, \quad \Theta \mapsto \Theta - g^{-1}(1 + \alpha) \rhd \Gamma \,,
\end{equation}
where the infinitesimal parameters \(\alpha \in C^{\infty}(M, \fg)\) and \(\Gamma \in \Omega^{1}(M, \fh)\) must satisfy
\begin{equation}
\label{ec:3dleftsymconst}
    \partial_z \alpha = 0\,, \quad \partial_z \Gamma_\tau = 0\,, \quad \partial_z \Gamma_{\bar z} = 0\,.
\end{equation}
Let us pause to highlight a few important observations. First, just like the fields in the theory, each left-/right-acting symmetry transformation parameter comes in a pair, once again reflecting the underlying categorical structure. Second, while the left-acting symmetry pair is anti-holomorphic, the right-acting symmetry satisfies a more intricate condition: it involves derivatives with respect to $\bar z$ and $\tau$, along with a non-trivial appearance of the differential $\mu_1$. 


The symmetry transformations \eqref{ec:3dleftsym} and \eqref{ec:3drightsym} 
give rise to conserved currents. Consistent with the higher--categorical 
structure of the theory, these currents again naturally appear in pairs. They can be 
derived using Noether’s trick; for the right-acting symmetry, one considers a variation of the action under $g \mapsto g(1 + \tilde\alpha) $ and $\Theta \mapsto \Theta - \tilde\alpha \rhd \Theta + \tilde\Gamma$ with arbitrary infinitesimal parameters \( \tilde\alpha \in C^\infty(M, \fg) \) and \( \tilde\Gamma \in \Omega^1(M, \fh) \). Requiring the action to be stationary under such variations, i.e., \( \delta S = 0 \), leads to 
\begin{equation}
\label{ec:conservationeq1}
     \partial_{\bar z}\tilde L_z -\mu_1\tilde H_{z{\bar z}}=0,\qquad \partial_{\tau}\tilde L
     _z-\mu_1\tilde H_{z\tau}=0,\qquad \partial_\tau H_{z\bar z} + \partial_{\bar z}H_{z\tau }=0\,.
\end{equation}
where we have introduced the currents
\begin{equation}
    \tilde L_z = g^{-1}\partial_z g \,,\quad \tilde H_{z{\bar z}} = \partial_z\Theta_{\bar z} + \mu_2(\tL_z,\Theta_{\bar z}),\quad \tilde H_{z\tau} =\partial_z\Theta_\tau+\mu_2(\tL_z,\Theta_\tau) \,.
\end{equation}
In a completely analogous manner, we can compute the currents associated to the left symmetry, given by
\begin{equation}
    \begin{aligned}
      &L_{\bar z} = -\partial_{\bar z}gg^{-1}-\mu_1(g\rhd \Theta_{\bar z})\,,\quad L_\tau = -\partial_\tau gg^{-1}-\mu_1(g\rhd \Theta_\tau)
      \,,\\
      &\qquad H_{{\bar z}\tau} = g\rhd (\partial_{\bar z}\Theta_\tau-\partial_\tau\Theta_{\bar z}-[\Theta_{\bar z},\Theta_\tau]) \,,
\end{aligned}
\end{equation}
which satisfy the conservation equations
\begin{equation}
\label{ec:conservationeq2}
    \partial_z L_{\bar z} = 0\,,\quad \partial_zL_\tau=0\,,\quad \partial_zH_{\bar z\tau} = 0\,.
\end{equation}

\subsection{Higher-currents as differential forms}
\label{sec:highcurrasdifforms}
The expressions we have obtained suggest a natural splitting of coordinates, treating \((\bar z, \tau)\) and \(z\) separately. More precisely, we can introduce the differential
\begin{equation}
    \mathrm{d}' = \bar{\partial} + \mathrm{d}_\tau\,,
\end{equation}
and decompose differential forms into components along \(\mathrm{d}'\) and along \(\partial\). With this in mind, the currents associated with the right-acting symmetry can be written as
\begin{equation}
\label{ec:currents1}
    \tilde{L} = g^{-1} \partial g\,, \quad \tilde{H} = \partial \Theta + \mu_2(g^{-1} \partial g, \Theta)\,,
\end{equation}
and those associated with the left-acting symmetry as
\begin{equation}
\label{ec:currents2}
    L = -\mathrm{d}'g\, g^{-1} - \mu_1(\Theta)\,, \quad 
    H = g \rhd \left( \mathrm{d}'\Theta - \tfrac{1}{2}[\Theta, \Theta] \right)\,.
\end{equation}
In terms of this decomposition, the conservation equation \eqref{ec:conservationeq1} takes the form
\begin{equation}
\label{ec:conseqdifform}
    \mathrm{d}'\tilde{L} - \mu_1(\tilde{H}) = 0\,, \quad \mathrm{d}'\tilde{H} = 0\,,
\end{equation}
and the differential constraints \eqref{ec:3dleftsymconst} satisfied by the symmetry parameters can be written as
\begin{equation}
    \dr' \tilde\alpha + \mu_1 (\tilde \Gamma) =0\,,\quad \dr'\tilde \Gamma = 0\,.
\end{equation}
Similarly, the conservation equation \eqref{ec:conservationeq2} can be written using the remaining differential as 
\begin{equation}
    \partial L = 0\,, \quad \partial H = 0\,,
\end{equation}
and the differential constraints \eqref{ec:3dleftsymconst} satisfied by the symmetry transformation parameters take the form
\begin{equation}
    \partial\alpha = 0 \,,\quad \partial  \Gamma = 0\,.
\end{equation}

\begin{remark}
    Let us note that this structure runs in direct parallel with the two dimensional WZW model. In that case, the theory admits independent holomorphic and anti-holomorphic symmetries which, using the Dolbeault decomposition $\dr = \partial + \bar\partial$, are characterized by parameters satisfying
\begin{equation}
    \bar\partial \tilde\alpha = 0 \,, \qquad \partial \alpha = 0\,.
\end{equation}
The associated currents are then defined by
\begin{equation}
    J = k g^{-1}\partial g \,, \qquad \bar J = k \bar\partial g\, g^{-1}\,,
\end{equation}
and satisfy the corresponding differential constraints
\begin{equation}
\label{ec:holocondition}
    \bar\partial J = 0 \,, \qquad \partial \bar J = 0\,.
\end{equation}
The equations satisfied by $\ta(z)$ and $J(z)$ imply that both define cohomology classes, and hence admit Laurent expansions. Indeed, the condition $\bar \partial \tilde \alpha = 0$ implies that\footnote{See Appendix~\ref{sec:cohomologyconventions} for our conventions on cohomological notation.}
\begin{equation}
    \tilde \alpha \in H^{(0,0)}_{\bar \partial}((\mathbb{C}\setminus\{0\})\otimes \fg) = \C[z,z^{-1}]\otimes \mathfrak{g}\, \quad \Longrightarrow \quad \tilde \alpha = \sum_{n\in \mathbb{Z}}\tilde \alpha_n z^n\,.
\end{equation}
Furthermore, using the linear isomorphism\footnote{We will make frequent use of this linear isomorphism. In two dimensions it 
amounts to writing $J = kg^{-1}\partial g$ as a one–form rather than as a 
$\fg$–valued function, a purely notational choice that eases the transition to 
the three–dimensional setting. Nevertheless, even in 3d it will at times be convenient 
to factor out the differential and regard these objects explicitly as 
either functions or forms.}
\begin{equation}
    \dr z \wedge -:H^{(0,0)}_{\bar \partial}((\mathbb{C}\setminus\{0\})\otimes \fg) \to H^{(1,0)}_{\bar \partial}((\mathbb{C}\setminus\{0\})\otimes \fg)\,,
\end{equation}
one can also write a Laurent expansion for the current $J$, since $\bar \partial J = 0$ implies 
\begin{equation}
    J \in H^{(1,0)}_{\bar \partial}((\mathbb{C}\setminus\{0\})\otimes \fg) \,.
\end{equation}
In other words, the mode expansions of the current and symmetry parameters, can be understood as arising from their realization as nontrivial Dolbeault cohomology classes.
\end{remark}

Turning back to our theory, our goal is to identify the infinite-dimensional symmetry algebra associated with the conserved currents 
$(\tilde L,\tilde H)$, $(L,H)$. For this, two ingredients are needed: mode expansions, which produce the infinitely many generators, and the short-distance behaviour of products of local operators, which determines the algebraic structure satisfied by these generators. Based on the preceding argument, a natural strategy is to work directly at the level of cohomology. In the present setting, this leads to the study of the complex \(\Omega^\bullet(\mathbb C\times\R\setminus\{0\})\) with respect to the splitting defined by \(\partial\) and \(\d'=\bar\partial+\d_\tau\). The natural framework for this is precisely the raviolo formalism of \cite{Garner:2023zqn}, to which we turn next.

\section{The Raviolo}\label{raviolo}

The cohomological viewpoint introduced above requires an appropriate local geometric model. In the two-dimensional holomorphic setting, this role is played by the punctured disk. On the one hand, the local behaviour of a holomorphic field near an insertion point is described on a punctured neighbourhood of that point, and holomorphic functions on such a neighbourhood admit Laurent expansions. On the other hand, for two local operators inserted at \(z_1\) and \(z_2\), translation invariance reduces their relative dependence to the coordinate \(z=z_1-z_2\in\mathbb C^\times\), so that near coincidence one is again led to a punctured neighbourhood of \(z=0\). In this sense, the punctured disk provides the common local geometry underlying both mode expansions and the short-distance behaviour of operator products, while the corresponding \(\bar\partial\)-cohomology captures the holomorphic data on this space.

In the topological-holomorphic setting on \(\mathbb C\times \mathbb R\), the analogous discussion leads to a different local geometric object. Consider two local operators inserted at points \((z_1,\tau_1)\) and \((z_2,\tau_2)\). After quotienting by overall translations, their relative position is described by $(z,\tau)=(z_1-z_2,\tau_1-\tau_2)$, with \((z,\tau)\neq(0,0)\). Thus the relevant local geometry is obtained from the space of nonzero relative separations in \(\mathbb C\times\mathbb R\). 

The crucial point, however, is that the two directions play different roles. In the holomorphic direction, one retains the full local dependence on the complex coordinate \(z\). By contrast, the real direction is topological, so configurations that differ only by a continuous deformation in the \(\tau\)-direction are identified, provided the two insertions remain distinct throughout the deformation. Thus, for each fixed \(z\neq 0\), all values of \(\tau\) are equivalent, since one may vary \(\tau\) freely without ever encountering the coincidence locus \((z,\tau)=(0,0)\).

The situation changes at \(z=0\). In that case, the configurations with \(\tau>0\) and \(\tau<0\) cannot be deformed into one another while keeping the insertions distinct, since any continuous path between them would necessarily pass through the coincidence point \((z,\tau)=(0,0)\). Thus the topological direction does not disappear completely: at \(z=0\) it retains a residual piece of ordering information, namely whether one insertion lies above or below the other. The resulting local model is therefore obtained by taking two copies of the disk, corresponding to these two orderings, and gluing them together away from the origin along the punctured disk. This is precisely the raviolo
\begin{equation}
    \mathbb{R}\mathrm{av} = D\cup_{D^\times}D\,.
\end{equation}
In the two-dimensional holomorphic case, the punctured disk provides the local model for \(\mathbb{C}^\times\), and \(\bar\partial\)-cohomology captures the holomorphic functions on it. In the topological-holomorphic case, the analogous local model is the raviolo, which plays the corresponding role for \(\mathbb{C}\times\mathbb{R}\setminus\{0\}\). The $\d'$-cohomology thus provides the topological-holomorphic analogue of holomorphic functions. 

The ideas sketched above are formalized in the raviolo framework of \cite{Garner:2023zqn}. We begin by reviewing the ingredients of this construction that will be needed in what follows.

\subsection{The mixed de Rham-Dolbeault complexes}
Let $M=\C\times \R$
with local coordinates 
$(z, \bar z, \tau)$. Since our discussion is entirely local, this model will be sufficient for our purposes. Although the formalism extends more generally to three-manifolds admitting a transverse holomorphic foliation, we will restrict to this case throughout.

In order to make the splitting of the total differential into $\d'$ and $\partial$ manifest, we introduce the bigraded complex 
\begin{equation}
    \mathcal{A}^{p,q}= C^{\infty}(M)\otimes \wedge^p (\C \d \tau \oplus \C \d \bar z)\otimes \wedge^q (\C \d z) \,.
\end{equation}
This complex will play an important role, so let us provide explicit examples of its elements for the lowest degrees. The space \(\mathcal{A}^{0,0}\) consists of smooth functions on $M$. An element of \(\mathcal{A}^{1,0}\) is of the form
\begin{equation}
    L_{\bar z}(z,\tau)\, \mathrm{d}\bar z + L_{\tau}(z,\tau)\, \mathrm{d}\tau\,.
\end{equation}
At degree 2, elements of \(\mathcal{A}^{1,1}\) are given by
\begin{equation}
    \tilde{H}_{z\bar z}(z,\tau)\, \mathrm{d}z \wedge \mathrm{d}\bar z + \tilde{H}_{z\tau}(z,\tau)\, \mathrm{d}z \wedge \mathrm{d}\tau\,,
\end{equation}
while elements of \(\mathcal{A}^{2,0}\) look like
\begin{equation}
    H_{\tau \bar z}(z,\tau)\, \mathrm{d}\tau \wedge \mathrm{d}\bar z\,.
\end{equation}
In particular, we have that
\begin{equation}
    \d':\cA^{p,q} \to \cA^{p+1,q} \,,\qquad \partial:\cA^{p,q} \to \cA^{p,q+1} \,.
\end{equation}

\medskip

\begin{remark}
\label{split}
Note that for \(q=0,1\), the complex \((\mathcal A^{\bullet,q},\d')\) may be identified
\begin{equation}
\label{ec:bicomplexA}
(\mathcal A^{\bullet,q},\d')
\cong
\bigl(\Omega^{(q,\bullet)}_{\mathrm{Dol}}(\mathbb C)\otimes \Omega^\bullet_{\mathrm{dR}}(\mathbb R),\ \bar\partial\otimes \mathbb 1+\mathbb 1\otimes \d_\tau\bigr).
\end{equation}
Likewise, one may consider the different complex
\begin{equation}
\label{ec:bicomplexB}
(\mathcal B^{\bullet,q},\partial+\bar\partial)
\cong
\bigl(\Omega^{(\bullet,\bullet)}_{\mathrm{Dol}}(\mathbb C)\otimes \Omega^q_{\mathrm{dR}}(\mathbb R),\ (\partial+\bar\partial)\otimes \mathbb 1\bigr),
\end{equation}
again for \(q=0,1\). In other words, the de Rham differential on $M$ admits two natural decompositions,
\begin{equation}
\d=\d'+\partial,
\qquad
\d=(\partial+\bar\partial)+\d_\tau.
\end{equation}
The complex \(\mathcal A\) is adapted to the first decomposition, whereas \(\mathcal B\) is adapted to the second. Although \(\mathcal A^{\bullet,q}\) and \(\mathcal B^{\bullet,q}\) are not quasi-isomorphic for fixed \(q\), one recovers the full complexified de Rham complex once the missing differentials are reinstated:
\begin{equation}
\left(\bigoplus_{q=0,1}(\mathcal A^{\bullet,q},\d'),\partial\right)
\cong
\Omega^\bullet_{\mathbb C}(M)
\cong
\left(\bigoplus_{q=0,1}(\mathcal B^{\bullet,q},\partial+\bar\partial),\d_\tau\right).
\end{equation}
This gives a precise notion of the \textit{chirality vector} $\ell$ introduced in \cite{Chen:2024axr}: it specifies which differential is ``missing". 
    
\end{remark}

\subsubsection{Tensoring with $\fG$}

An important structural feature of our three-dimensional theory \eqref{3daction-coords} is that the currents $(\tL,\tH)$ and $(L,H)$  defined in \eqref{ec:currents1} and \eqref{ec:currents2}, respectively, are valued in the Lie 2-algebra \(\mathfrak G=(\mathfrak g\xrightarrow{\mu_1}\mathfrak h,\mu_2)\). Because \(\mathfrak G\) comes equipped with its own differential \(\mu_1\), it must be regarded as a chain complex. Hence, to describe the objects appearing in our theory, we are naturally led to consider the tensor product complex in which the geometric differential \(\d'\) is combined with the internal differential \(\mu_1\), namely
\begin{equation}
    \begin{tikzcd}
  \cdots \arrow[r] & \cA^{0,q} \otimes \fh \arrow[r, "\dr'"] \arrow[d, "\mu_1"] & \cA^{1,q} \otimes \fh \arrow[r] \arrow[d, "\mu_1"] &\cdots \\
  \cdots \arrow[r] & \cA^{0,q} \otimes \fg\arrow[r, "\dr'"] & \cA^{1,q} \otimes \fg  \arrow[r] & \cdots
\end{tikzcd}
\end{equation}
for $q=0,1$. As a tensor product complex $\cA^{\bullet,q}\otimes \fG$ is equipped with the differential
\begin{equation}
    \hat \dr' = \dr'\otimes \mathbb{1} - (-1)^{\text{deg}} \mathbb{1}\otimes \mu_1\,,
\end{equation}
where ``$\text{deg}$'' denotes the homogeneous degree of the first entry. More explicitly, for each $\alpha\otimes\mathsf{x}\in \cA^{\bullet,q}\otimes \fG$ we have 
\begin{equation}
    \hat \dr'(\alpha\otimes \mathsf{x}) = \dr'\alpha\otimes \mathsf{x} - (-1)^{|\alpha|} \alpha\otimes \mu_1(\mathsf{x}) \,,
\end{equation}
where $|\alpha|$ is the homogeneous degree of $\alpha\in\cA^{\bullet,q}$.

\subsection{Revisiting the Three-Dimensional Theory}

Let us now apply the framework developed above to reformulate our 3d topological-holomorphic field theory \(S[g,\Theta] \) in terms of the tensor product complex \( \mathcal{A}^{p,q} \otimes \fG \).

We begin by noting that, since the component $\Theta_z$ is missing from the action \eqref{3daction-coords}, we can view the fundamental fields of $S[g,\Theta]$ as elements
\begin{equation*}
    g\in \cA^{0,0}\otimes G,\quad \Theta\in \cA^{1,0}\otimes \fh\,.
\end{equation*}
This then allows us to write $S[g,\Theta]$ in  a very compact form
\begin{equation}
  S[g,\Theta]= \int_{M}\langle \partial (\dr'gg^{-1}),g\rhd \Theta\rangle - \tfrac{1}{2}\langle \mu_1(\partial (g\rhd\Theta)),g\rhd\Theta\rangle\,.\label{3daction}
\end{equation}
The symmetry transformation parameters of \eqref{ec:3dleftsym} and \eqref{ec:3drightsym} can also be identified with elements of the complex as 
\begin{equation}
    \alpha,\tilde \alpha \in \cA^{0,0}\otimes \fg \,,\quad \Gamma,\tilde\Gamma \in \cA^{1,0}\otimes \fh \,,
\end{equation}
as well as the currents \eqref{ec:currents1} and \eqref{ec:currents2} which can be  written as
\begin{align}
        & L \in \mathcal{A}^{1,0}\otimes \fg\,,\quad H\in\mathcal{A}^{2,0}\otimes \fh\\
        & \tilde L \in\mathcal{A}^{0,1}\otimes \fg\,,\quad \tilde H \in\mathcal{A}^{1,1}\otimes \fh\,.
\end{align}
The conservation equations for the currents \((\tilde{L}, \tilde{H})\) given in \eqref{ec:conseqdifform}, associated with the right symmetry of \eqref{3daction}, are then given by 
\begin{equation}
\label{ec:currentsincoho}
    \hat\dr' (\tilde L, \tilde H) = (\dr'\tilde L - \mu_1 (\tilde H)
    , \dr' \tilde H) = 0 \,,
\end{equation}
so that $(\tilde L,\tilde H)$ is a representative of a cohomology class in $ H_{\hat \dr'}^{(\bullet,1)}(M\otimes \fG)$\footnote{We refer the reader to appendix \ref{sec:cohomologyconventions} for our conventions on cohomological notation.}. Note that $\mu_1$ does not act on the second factor $\tilde H \in \cA^{1,1}\otimes \fh$ as it is already valued in $\fh$. Similarly, the differential constraints \eqref{ec:3drightsymconst} satisfied by the symmetry transformation parameters corresponding to the right symmetry can be written as
\begin{equation}
\label{ec:paramincoho}
    \hat \dr' (\tilde \alpha,\tilde \Gamma) = (\dr'\tilde \alpha + \mu_1(\tilde \Gamma),\dr' \tilde \Gamma)=0 \,,
\end{equation}
so that $(\tilde \alpha, \tilde \Gamma)$ is a representative of a cohomology class in $H^{(\bullet,0)}_{\hat \dr'}(M\otimes \fG)$. 

Conversely, the conservation equations for the currents $(L,H)$ given in \eqref{ec:3dleftsym} associated to the left symmetry of \eqref{3daction} are given by 
\begin{equation}
    \partial (L,H) = (\partial L, \partial H)=0 \,,
\end{equation}
so that $(L,H)$ is a representative of a cohomology class in $H^{(\bullet,1)}_\partial(M\otimes \fG)$, whereas the differential constraints \eqref{ec:3dleftsymconst} take the form
\begin{equation}
    \partial (\alpha,\Gamma)=(\partial \alpha,\partial\Gamma)=0 \,,
\end{equation}
so that $(\alpha,\Gamma)$ is a representative of a cohomology class in $H^{(\bullet,0)}_\partial(M\otimes \fG)$. 
It thus follows that the higher-currents of our 3d theory admit a natural decomposition, captured by the mixed de Rham--Dolbeault complex \(\mathcal{A} \otimes \mathfrak{G}\).  We shall call $H_{\hat \dr'}^{(\bullet,\bullet)}$ the \textbf{chiral sector}, and $H_{\partial}^{(\bullet,\bullet)}$ the \textbf{anti-chiral sector} so that $(\tilde L,\tilde H)$ are chiral  whereas $(L,H)$ are anti-chiral.

It is important to note that in this case, the chiral and anti-chiral sectors are not ``symmetric'' to one another in an obvious way. In particular, as we will see, only the chiral sector will give rise to the infinite dimensional symmetry algebra. Consequently, we will focus on the cohomology group
\begin{equation}
    H_{\hat{\mathrm{d}}'}^{(\bullet,q)}(M \otimes \mathfrak{G})\,,
\end{equation}
for $q=0,1$, with the objective of finding a mode expansion for both our higher current $(\tilde L, \tilde H)$ and the gauge transformation parameters $(\tilde \alpha,\tilde \Gamma)$. 



\begin{remark}\label{nonchiral}

In \cite{Chen:2024axr,Chen:2023integrable}, a non-chiral version of the three-dimensional field theory was also considered. This variant is formulated using the decomposition of the de~Rham differential associated with the bicomplex \(\CB^{\bullet,\bullet}\) defined in \eqref{ec:bicomplexB}, given by the splitting
\begin{equation}
        \mathrm{d} = (\partial + \bar{\partial}) + \mathrm{d}_\tau\,.
\end{equation}
The fields consist of a group-valued field \(g \in \CB^{0,0} \otimes G\) and a one-form field \(\Theta \in \CB^{1,0} \otimes \mathfrak{h}\). The corresponding action is given by
\begin{equation}
\label{ec:nonchiralaction}
      S_{\text{nc}}[g,\Theta] = \int_M \left\langle \mathrm{d}_\tau(\mathrm{d}_\Sigma g g^{-1}), g \rhd \Theta \right\rangle 
    - \tfrac{1}{2} \left\langle \mu_1\big(\mathrm{d}_\tau(g \rhd \Theta)\big), g \rhd \Theta \right\rangle\,, 
\end{equation}
where we have written $\dr_\Sigma = \partial + \bar \partial$ for short. This theory also features left- and right-acting symmetries, with associated conserved currents \((L_{\text{nc}}, H_{\text{nc}})\) and \((\tilde{L}_{\text{nc}}, \tilde{H}_{\text{nc}})\), which satisfy explicitly \textit{non-chiral} conservation equations (see equation (6.3) in~\cite{Chen:2023integrable}):
\begin{equation*}
    \dr_\tau(L_{\text{nc}}, H_{\text{nc}}) = 0\,, \qquad (\dr_\Sigma - \mu_1)(\tilde{L}_{\text{nc}}, \tilde{H}_{\text{nc}}) = 0\,.
\end{equation*}
This highlights a fundamental difference between the non-chiral theory \(S_{\mathrm{nc}}\) and its chiral counterpart: their dynamics are fundamentally different --- indeed, \(S_{\mathrm{nc}}\) is fully topological (see \textbf{Theorem}~6.2 in~\cite{Chen:2024axr}).

\end{remark}

\subsection{Mode expansion of the currents}\label{ravpolys}

We have seen that the differential constraints for $(\tilde \alpha,\tilde\Gamma)$ and the conservation equations for \((\tilde{L}, \tilde{H})\) are related to the cohomology group
\begin{equation}
    H_{\hat{\mathrm{d}}'}^{(\bullet,q)}(M \otimes \mathfrak{G})\,,
\end{equation}
with $q=0,1$. Our next goal is to obtain an explicit expression for this cohomology group.
We will proceed in two steps: first, we will describe 
\(H^{(\bullet,q)}_{\mathrm{d}'}(M)\) following \cite{Garner:2023zqn}, 
and then we will apply the Künneth formula to compute the cohomology of 
the tensor product complex. 

Similar to the two-dimensional pure Dolbeault case, there is a linear isomorphism 
\begin{equation}
\label{ec:lineariso}
    \dr z \wedge -: H^{(\bullet,0)}_{\mathrm{d}'}(M) \to H^{(\bullet,1)}_{\mathrm{d}'}(M) \,.
\end{equation}
Hence, it suffices to find an explicit realization of \( H^{(\bullet,0)}_{\mathrm{d}'}(M) \), since elements in \( H^{(\bullet,1)}_{\mathrm{d}'}(M) \) can be obtained by wedging elements of the former with \(\mathrm{d} z\). The cohomology of \((\mathcal{A}^{\bullet,0}, \mathrm{d}')\) was completely characterized in \cite{Garner:2023zqn} and is concentrated in degrees 
zero and one. The zeroth cohomology consists of polynomials in the holomorphic 
variable \(z\), that is,
\begin{equation}
    H^{(0,0)}_{\mathrm{d}'}(M) = \mathbb{C}[z]\,.
\end{equation}
Unlike in the Dolbeault complex, no negative powers of \(z\) appear here; this is a manifestation of Hartogs's theorem (see also \cite{FAONTE2019389}). 
Consequently, any element \(\tilde{\alpha} \in H^{(0,0)}_{\mathrm{d}'}(M)\) can be 
expanded in modes as
\begin{equation}
    \tilde{\alpha} = \sum_{n=0}^\infty \tilde{\alpha}_n z^n\,.
\end{equation}
The degree-one cohomology \(H^{(1,0)}_{\mathrm{d}'}(M)\) is more subtle and 
requires introducing additional structures. Since the negative powers of 
\(z\) that appear in the holomorphic setting are absent from 
\(H^{(0,0)}_{\mathrm{d}'}(M)\), the idea is that in this context, 
the role of \(z^{-1}\) is played by a different object with analogous properties, 
which appears in \(H^{(1,0)}_{\mathrm{d}'}(M)\). Concretely, given a point $\underline{z} = (z,\bar z,\tau) \in  \C \times \R$, this is given 
by the \((1,0)\)-form
\begin{equation}
\label{ec:omega1}
    \omega(\underline{z}) = \frac{\tau \, \dr \bar{z} - 2 \bar{z} \, \dr \tau}{(|z|^2 + \tau^2)^3} 
    \in \C\d\tau\oplus\C\d\bar z\,.
\end{equation}
It is straightforward to verify that $\dr' \omega = 0$; however, $\omega$ is
not $\dr'$-exact and thus defines a non-trivial cohomology class in \(H^{(1,0)}_{\mathrm{d}'}(M)\). By contrast,
$\partial_\tau \omega(\underline{z})$ and $\partial_{\bar z} \omega(\underline{z})$
are $\dr'$-exact, and therefore trivial in cohomology (see Lemma~1.1.3
in~\cite{Garner:2023zqn}).

The key property 
of \(\omega\) is an analogue of Cauchy's residue theorem: its integral 
over a two-sphere centered at the origin is given by
\begin{equation}
    \oint_{S^2}\dr z \wedge \omega = 8\pi i\,.
\end{equation}
Moreover, just as \( z^{-1} \) is the Green's function for \(\partial\) on 
\(\mathbb{C} \setminus \{0\}\), the \((1,0)\)-form \(\omega\) serves as the 
Green's function for \(\dr'\) on \(\C \times \R  \setminus \{0\}\). 
Consequently, they both act as propagators in field theories with the corresponding kinetic operators.

We can now use \(\omega\) to give an explicit description of 
\(H^{(1,0)}_{\mathrm{d}'}(M)\). First, we introduce the so-called degree-one 
\emph{raviolo differential forms}, which are obtained by taking derivatives of 
\(\omega\) with respect to \(z\). Specifically, we define
\begin{equation}
\label{ec:Omegadef}
    \Omega^m = \frac{(-1)^m}{m!} \partial_z^m \omega\,,\qquad \Omega^0=\omega\,.
\end{equation}
Each \(\Omega^m\) is a one-form, and they obey the relations
\begin{equation}
\label{ec:zntimesomegam}
    z^n \Omega^m = 
    \begin{cases}
        0 & \text{for } n > m \\
        \Omega^{m-n} & \text{for } n \leq m
    \end{cases}\,,
\end{equation}
together with the holomorphic derivative identity 
\begin{equation}
\label{ec:holoderivative} 
    \partial_z \Omega^m = -(m+1)\Omega^{m+1}\,.
\end{equation}
In other words, the raviolo differential forms serve a role analogous to that of negative modes in $\C \times \R\setminus\{0\}$, but in the form of \textit{1-forms} with non-trivial degree. 

The degree-one cohomology is then given by \cite{Garner:2023zqn}
\begin{equation}
    H^{(1,0)}_{\mathrm{d}'}(M) = \mathrm{span}_{\mathbb{C}} \{ \Omega^0, \Omega^1, \Omega^2, \dots \}\,,
\end{equation}
so that any element \(\tilde{\Gamma} \in H^{(1,0)}_{\mathrm{d}'}(M)\) can be 
expanded in modes as
\begin{equation}
    \tilde{\Gamma} = \sum_{m=0}^\infty \tilde{\Gamma}_m \Omega^m\,.
\end{equation}
We will denote the full cohomology by
\begin{equation}
\label{ec:ravpoly}
    \mathcal{K}^\bullet_{\text{poly}} := H^{(\bullet,0)}_{\mathrm{d}'}(M)\,,
\end{equation}
and refer to it as the space of polynomials on the raviolo. This space is the 
analogue of the Laurent polynomials \(\mathbb{C}[z, z^{-1}]\) on 
\(\mathbb{C}^\times\) and provides the appropriate algebraic structure for 
constructing the mode expansions, ultimately leading to the centrally extended affine graded Lie algebra. 

Having found an explicit realization for \( H^{(\bullet,0)}_{\mathrm{d}'}(M) \), we can use the linear isomorphism in \eqref{ec:lineariso} to construct elements in \( H^{(\bullet,1)}_{\mathrm{d}'}(M) \). The existence of this isomorphism implies that any element \( \tilde{L} \in H^{(0,1)}_{\mathrm{d}'}(M) \) can be written as \( \tilde{L} = \tilde{L}_z \mathrm{d} z \) for some \( \tilde{L}_z \in H^{(0,0)}_{\mathrm{d}'}(M) \) and similarly with elements in $H^{(1,1)}_{\mathrm{d}'}(M)$. With a slight abuse of notation we will write
\begin{equation}
    \mathrm{d}z\wedge\mathcal{K}_{\mathrm{poly}}^\bullet  := H^{(\bullet,1)}_{\mathrm{d}'}(M) \,.
\end{equation}

With an explicit expression for both \( H^{(\bullet,0)}_{\mathrm{d}'}(M) \) and \( H^{(\bullet,1)}_{\mathrm{d}'}(M) \) at hand, we can proceed to determine the cohomology of the tensor product complex \( H_{\hat{\mathrm{d}}'}^{(\bullet,q)}(M \otimes \mathfrak{G}) \). To do so, we make use of the \emph{Künneth formula} which relates the cohomology of the tensor product with the tensor product of the cohomologies
\begin{equation}
    H_{\hat \dr'}^{(0,q)}(M\otimes \fG)\cong\bigoplus_{0=m+k}H^{(m,q)}_{\dr'}(M)\otimes H^k_{\mu_1}(\fG) \,,
\end{equation}
where the "correction" to the K{\"u}nneth formula, which are the derived Ext-groups, do not appear here as both of our complexes $\CA,\fG$ are free.

Since $\mathfrak{G}$ is concentrated in degrees $(-1)$ and $0$, its cohomology will also be concentrated in these degrees. Let $H^\bullet(\mathfrak{G}) = V \oplus \mathfrak{n}$
denote the cohomology of the Lie 2-algebra $\mathfrak{G}$, with
\begin{equation}
    V = \ker \mu_1 = H^{-1}(\mathfrak{G}), \qquad \mathfrak{n} = \mathfrak{g}/\mathrm{im}(\mu_1)= H^0(\mathfrak{G})\,,
\end{equation}
where $V$ is an Abelian $\mathfrak{n}$-module. By a theorem of Gerstenhaber, the cohomology $H^\bullet(\fG)$ is part of the data which characterizes Lie 2-algebras $\fG$ up to equivalence \cite{Wagemann+2021}. In the special case where $\mu_1=0$, namely when $\fG$ is called \textit{skeletal}, then $\fG = H^\bullet(\fG)$ is its own cohomology. On the other hand, if $\mu_1=\id$ is the identity, then $H^\bullet(\fG)=0$ is trivial.

We then have the following proposition.

\begin{prop}\label{raviolopoly}
The symmetry transformation parameters $(\tilde \alpha,\tilde \Gamma)$ corresponding to the right symmetry of the action \eqref{3daction} determine an element of the degree-0 cohomology group
\begin{equation}
    (\tilde \alpha,\tilde \Gamma) \in H^{(0,0)}_{\hat \dr'}(M\otimes \mathfrak{G}) = (\mathcal{K}_{\mathrm{poly}}^0 \otimes \fn ) \oplus (\mathcal{K}^1_{\mathrm{poly}} \otimes V) \,.
\end{equation}
Similarly, the higher currents $(\tilde L,\tilde H)$ corresponding to this symmetry determine an element of the degree-1 cohomology group 
\begin{equation}
    (\tilde L,\tilde H)\in H_{\hat \dr'}^{(0,1)}(M\otimes \fG)=((\dr z \wedge \mathcal{K}_{\mathrm{poly}}^0) \otimes \fn)\oplus  ((\dr z\wedge\mathcal{K}_{\mathrm{poly}}^1 ) \otimes V)
\end{equation}

 \end{prop}

 \begin{proof}
     The proof follows immediately from equations \eqref{ec:currentsincoho} and \eqref{ec:paramincoho} 
\end{proof}

The fundamental consequence of this proposition is that it provides a mode expansion for both our symmetry transformation parameters and higher currents. Specifically, we can express them as
\begin{equation}
\label{ec:paramodes}
    \tilde{\alpha} = \sum_{n=0}^\infty \tilde{\alpha}_n z^n \,,\quad \tilde{\Gamma} = \sum_{m=0}^\infty \tilde{\Gamma}_m \Omega^m
\end{equation}
with $\tilde{\alpha}_n \in \mathfrak{n}$ and $\tilde{\Gamma}_m \in V$. Similarly, the higher currents admit the expansions
\begin{equation}
\label{ec:curmodes}
    \tilde{L} = \sum_{n=0}^\infty \tilde{L}_n z^n \,\dr z\,,\quad \tilde{H} = \dr z \wedge\sum_{m=0}^\infty \tilde{H}_m  \Omega^m
\end{equation}
where $\tilde{L}_n \in  \fn$ are $\mathfrak{n}$ and $\tilde{H}_m \in V$.

The aim of the following sections is to study in detail the current algebra satisfied by these currents, and we shall do so in $\hat{\d}'$-cohomology. That is, we restrict attention to $\hat{\d}'$-closed (ie. on-shell) current operators modulo $\hat{\d}'$-exact terms. The resulting algebra is thus the one induced on the cohomology classes, or the \textit{cohomological descendant operators}. This is the natural framework for our purposes, since the explicit structure of the local $\hat{\d}'$-cohomology provides the expansion basis with which these currents can be analysed.


\section{Radial Quantisation}\label{affineraviolo}

With the mode expansions of the currents and transformation parameters in place, we can compute the commutation relations of the modes of $(\tilde L, \tilde H)$ and thereby determine the current algebra. The logic is directly analogous to the familiar two-dimensional case: singular terms in the operator product expansion determine the commutators of the corresponding modes once one passes to radial quantisation.

The reason this construction extends naturally to $\C\times\R$ is that the two ingredients underlying the two-dimensional argument also admit higher-dimensional analogues. The first is the mode expansion of the currents, established in the previous section. The second is Cauchy's residue theorem, whose role in two dimensions is to relate the singular part of the OPE to the commutation relations of the modes. In the present setting, the corresponding role is played by its higher-dimensional analogue, formulated in terms of the Bochner--Martinelli $(1,0)$-form $\omega$ introduced earlier in \eqref{ec:omega1}, which satisfies
\begin{equation}
\label{ec:intomegaes8pi}
    \oint_{S^2} \d z \wedge \omega = 8\pi i \, .
\end{equation}
This replaces contour integrals over circles in $\C$ by integrals over spheres in $\C\times\R$. In Euclidean signature, one may again take the radial coordinate as time and define a radial quantisation scheme. Within this framework, the generalised residue formula allows one to translate the singular part of the three-dimensional OPE into commutation relations of the current modes.

\subsection{The Ward Identity}
\label{sec:wardidentitiy}
To set up the computation of the Ward identity, we first fix some conventions. As preempted in \S\ref{sec:highcurrasdifforms}, from this point onward it is convenient to suppress the explicit $\d z$ appearing in the definition of the currents $(\tL,\tH)$. Equivalently, we contract the currents with the vector field $\partial_z$. This is entirely analogous to the two-dimensional case, where the current $J(z)=g^{-1}\partial_z g$ is viewed as a $\fg$-valued function rather than as a $1$-form. With a slight abuse of notation, we continue to denote these contractions by $(\tL,\tH)$, so that $\tL$ is now a $\fg$-valued function and $\tH$ an $\fh$-valued $(1,0)$-form.

The goal is to compute the Ward Identity associated with the right acting transformation 
\begin{equation}
\label{ec:ratransf}
   g\to g+g\ta\,,\quad \Theta \to \Theta-\ta\rhd \Theta+\tg\,.
\end{equation}
Hence, we consider $\ta$ and $\tg$ with compact support in a ball $B \subset \C\times\R$, vanishing outside $B$, and subject to the differential constraints
\begin{equation}
\label{ec:difconst}
   \hat \d'(\ta,\tg)=(\d'\ta + \mu_1\tg,\d'\tg) = 0\,.
\end{equation}
Under this transformation, the action \eqref{3daction} varies as
\begin{equation}
    \delta_{(\ta,\tg)}S = \oint_{\partial B}\d z \wedge (\langle \ta,\tH\rangle + \langle \tL,\tg\rangle) \,,
\end{equation}
where we have used \eqref{ec:difconst} together with Stokes' theorem.

Since we are working in $\hat{\d}'$-cohomology, we must first check that the algebraic operations appearing in the variation above are well defined on cohomology classes. In appendix \ref{sec:appendixproof}, we show that the bracket, the action, and the bilinear form descend to cohomology. From now on, we use the same notation $[-,-]$, $\mu_2$, and $\langle-,-\rangle$ for the induced bracket, action, and bilinear form on $H^\bullet(\fG)$. Moreover, given a basis \( \{t_a\} \) of \( \mathfrak{n} \) and \( \{s_b\} \) of \( V \), we shall use the same notation for the corresponding structure constants and bilinear form, whenever no confusion can arise:
\begin{equation}
    [t_a, t_b] = f_{ab}^c\, t_c\,, \qquad 
    \mu_2(t_a, s_b) = (\mu_2)_{ab}^c\, s_c\,, \qquad 
    \langle t_a, s_b \rangle = \kappa_{ab}\,.
\end{equation}
This allows us to express the currents and transformation parameters as linear combinations of the generators of the algebras,
\begin{equation}
    \ta= \sum_{a=1}^{\dim \fn}\ta^a t_a \,,\quad \tH = \sum_{b=1}^{\dim V}\tH^b s_b \,, \quad \tL = \sum_{a=1}^{\dim \fn}\tL^a t_a\,,\quad \tg=\sum_{b=1}^{\dim V}\tg^b s_b \,,
\end{equation}
and the Ward identity takes the form\footnote{Note that the symbol $\langle-\rangle$ denotes a correlation function, while the symbol $\langle-,-\rangle$ refers to the bilinear form on $\fG$. 
}
\begin{equation}
    \delta_{(\ta,\tg)}\langle X\rangle = \oint_{\partial B} \dr z\wedge \left(\ta^a\langle \tH^b X\rangle + \langle \tL^aX\rangle \tg_b\right)\kappa_{ab}\,.
\end{equation}
We can now use this relation to determine the operator product expansion between the currents. 

In doing so, it is convenient to make one further remark on notation. Although the theory is defined on $\C\times\R$, in what's next we shall suppress the $\tau$-coordinate and write the currents as depending only on $z$. This is possible because we are working in $\hat{\d}'$-cohomology, where the $\tau$-dependence of the relevant representatives is cohomologically trivial. This can be seen directly from the expansions of the currents in \eqref{ec:curmodes}: $\tL$ is manifestly a function of $z$ alone, while $\tH$ is expanded in the $\Omega$-basis, whose $\tau$-derivative is $\d'$-exact. Notably, this should be understood as a statement about cohomology classes, not about the underlying geometry: the $\tau$-direction still plays a role in the three-dimensional setting, but the operator product expansions below are statements in cohomology.

To derive the OPEs, we will make use of the following lemma.
\begin{lemma}
    Let $S_{w}^2$ denote a sphere centered at $w$. Then
    \begin{equation}
    \label{ec:gencauchyform}
        \frac{1}{8\pi i}\oint_{S_{w}^2}\d z \wedge (z-w)^n\Omega_{z-w}^m = \delta_{m,n}
    \end{equation}
where $\Omega_{z-w}^m = \frac{(-1)^m}{m!} \partial_z^m \omega(z-w)$ for $\omega$ defined in \eqref{ec:omega1}.
\end{lemma}

\begin{proof}
   The proof is an immediate consequence of the relation \eqref{ec:zntimesomegam}, together with the integral of $\omega$ over a sphere \eqref{ec:intomegaes8pi}.  
\end{proof}

We thus have the following proposition
\begin{prop}
\label{thm:opes}
    The currents $\tL$ and $\tH$ have the following operator product expansions
    \begin{align}
        & \tL^a(z)\tL^b(w) = \mathrm{reg}
        \label{ec:LLOPE}\\
        & \tH^a(z)\tL^b(w) = \Omega_{z-w}^0 (\tilde \mu_2)^{ab}_c \tL^c(w) + \Omega_{z-w}^1(\kappa^{-1})^{ab}+\mathrm{reg} \label{ec:HLOPE} \\
        &\tH^a(z)\tH^b(w) = \Omega_{z-w}^0 \tilde f^{ab}_c \tH^c(w) + \mathrm{reg}
        \label{ec:HHOPE}
    \end{align}
where $\mathrm{reg}$ stand for terms which are regular as $z \to w$, and where we have defined the dual structure constants
\begin{equation}
        \tilde f^{ab}_c = (\kappa^{-1})^{ax}(\kappa^{-1})^{by}f_{xy}^z\kappa_{zc},\qquad  (\tilde\mu_2)^{ab}_c = (\kappa^{-1})^{ax}\kappa_{cz}(\mu_2)_{xy}^z(\kappa^{-1})^{yb} \,.\label{dualbrackets}
    \end{equation}    
\end{prop}
\begin{remark}
Note that the bilinear form \(\langle \cdot,\cdot \rangle : \mathfrak{n} \times V \to \mathbb{C}\) pairs elements of \(\mathfrak{n}\) with those of \(V\), and thus, unlike the standard case, it does not provide an identification between each of these spaces with their duals individually. As a result, we cannot use \(\kappa\) to raise or lower indices; this is why we explicitly defined the dual structure constants $\tilde f^{ab}_c$ and $(\tilde{\mu}_2)_c^{ab}$. 
\end{remark}

\begin{proof}
   We begin by computing the variations of the currents $\tL$ and $\tH$ under the symmetry transformation \eqref{ec:ratransf}, which are given by
    \begin{align}
    \delta_{(\tilde \alpha, \tilde \Gamma)} \tilde{L} &= [\tilde L, \ta] + \partial_z \tilde \alpha \,, \\
    \delta_{(\tilde \alpha, \tilde \Gamma)} \tilde{H} &= -\mu_2(\tilde \alpha, \tilde{H}) + \mu_2(\tilde{L}, \tilde \Gamma) + \partial \tilde \Gamma \,.
\end{align}
taking the product of fields to be $X = \tL^b(w)X'$, with 
$X'=\prod_i \mathcal O_i(u_i)$ where $w \in B$ and all $u_i \notin B$, where the $\mathcal O_i$ denote arbitrary operator insertions, the vanishing of $\ta$ and $\tg$ outside $B$ implies that $\delta_{\ta,\tg}$ acts only on $\tL^b(w)$. The Ward identity takes the form
\begin{multline}
\label{ec:Wardidope}
    \langle ([\tilde L,\ta]^b+\partial_z\tilde\alpha^b)X'\rangle
    =\frac{1}{8\pi i}\oint_{\partial B}\d z\wedge \ta^c\,
        \langle \tH^a(z)\tL^b(w)X'\rangle \kappa_{ca} \\
    +\frac{1}{8\pi i}\oint_{\partial B}\d z\wedge 
        \langle \tL^a(z)\tL^b(w)X'\rangle \tg^c \kappa_{ca}\,.
\end{multline}
Since the left-hand side contains no terms proportional to $\tg$, we conclude
\begin{equation}
    \tL^a(z)\tL^b(w)=\mathrm{reg}\,.
\end{equation}

Let us now check that the OPE \eqref{ec:HLOPE} reproduces the remaining terms. Inserting the first contribution from \eqref{ec:HLOPE} in \eqref{ec:Wardidope} yields
\begin{equation}
\label{ec:OPEterm1a}
    \frac{1}{8\pi i}\oint_{\partial B}\d z\wedge \ta^c \Omega_{z-w}^0 
        (\tilde\mu_2)^{ab}_d \langle \tL^d(w)X'\rangle \kappa_{ca}
    =\ta^c(\tilde\mu_2)^{ab}_d \langle \tL^d(w)X'\rangle \kappa_{ca}\,,
\end{equation}
where we have used the generalized Cauchy formula \eqref{ec:gencauchyform}.  
On the other hand, the first term on the left-hand side of 
\eqref{ec:Wardidope} reads in components
\begin{equation}
\label{ec:OPEterm1b}
    \ta^c f^{b}_{dc}\langle \tL^d(w)X'\rangle\,.
\end{equation}

To compare the two expressions we invoke the ad-invariance of the bilinear form,
\begin{equation}
    \langle [t_a,t_b],s_c\rangle = \langle t_a,\mu_2(t_b,s_c)\rangle 
    \quad\Rightarrow\quad 
    f_{ab}^d \kappa_{dc} = \kappa_{ad}(\mu_2)_{bc}^d\,,
\end{equation}
which implies
\begin{equation}
    (\tilde\mu_2)^{ab}_d \kappa_{ca}  
    = (\kappa^{-1})^{ax}\kappa_{dz}(\mu_2)_{xy}^z(\kappa^{-1})^{yb}\kappa_{ca}
    = f^z_{dc}\kappa_{zy}(\kappa^{-1})^{yb}= f^b_{dc}\,,
\end{equation}
where in the first line equality used the definition of $(\tilde\mu_2)^{ab}_d$ given in \eqref{dualbrackets}, and in the third, the $\mathrm{ad}$-invariance of the bilinear form. Hence the contributions \eqref{ec:OPEterm1a} and \eqref{ec:OPEterm1b} agree. Similarly, if we insert the second contribution from \eqref{ec:HLOPE} in \eqref{ec:Wardidope} we find
\begin{equation}
 \begin{split}
\label{ec:OPEterm2a}
    \frac{1}{8\pi i}\oint_{\partial B}\d z\wedge \ta^c \Omega_{z-w}^1 
        (\kappa^{-1})^{ab} \langle X'\rangle \kappa_{ca}
    & =\frac{1}{8\pi i}\oint_{\partial B}\d z\wedge \partial_z\ta^b \Omega_{z-w}^0 
         \langle X'\rangle \\
         &= \partial_z \ta^b\langle X'\rangle
\end{split}   
\end{equation}
where we used the holomorphic derivative relation $\Omega_{z-w}^1=-\partial_z\Omega_{z-w}^0$ 
to integrate by parts in the first step, and then applied the generalized Cauchy
formula \eqref{ec:gencauchyform} in the second. This reproduces exactly the 
second term on the left-hand side of \eqref{ec:Wardidope}.

With a completely analogous computation, one verifies that the $\tH \tH$-OPE takes the form \eqref{ec:HHOPE}.

\end{proof}



\subsection{Commutation Relations}\label{sec:commutators}
As we mentioned at the beginning of this section, the essence of this construction is that, in analogy with the two-dimensional case, one may again perform radial quantisation. Since the theory is defined on $\C\times\R \cong \mathbb{R}^3$, we choose as time coordinate the radial variable $r(\underline{z})=\sqrt{\tau^2+|z|^2}$, where we recall, $\underline{z}=(z,\bar z,\tau)$. Equal-time slices are therefore spheres of fixed radius. Correlation functions are then identified with vacuum expectation values of radially ordered products of operators inserted at points $\underline{z},\underline{w}\in\C \times \R$,
\begin{equation}
    \langle a(\underline{z})b(\underline{w})\rangle
    = \langle 0|\mathcal{R}(a(\underline{z})b(\underline{w}))|0\rangle \,,
\end{equation}
where the radial-ordering operator $\mathcal{R}$ is defined by
\begin{equation}
    \mathcal{R}(a(\underline{z})b(\underline{w}))=
    \begin{cases}
        a(\underline{z})b(\underline{w}) & \text{if } r(\underline{z})>r(\underline{w})\,,\\
        b(\underline{w})a(\underline{z}) & \text{if } r(\underline{w})>r(\underline{z})\,.
    \end{cases}
\end{equation}
We now use radial ordering to convert the OPEs established in Proposition \ref{thm:opes} into a commutator formula. For the currents under consideration, the OPEs are either regular or have singular part proportional to the kernel $\Omega_{\underline{z}-\underline{w}}$. In the latter case, the singular contribution to the radially ordered product defines an operator-valued 1-form in the integration variable $\underline{z}$. Wedging with $\d z$ then yields an operator-valued $2$-form, which may be integrated over a small sphere surrounding the insertion at $\underline{w}$.

We now use radial ordering to evaluate the integral
\begin{equation}
    \oint_{S_{\underline{w}}^2}\d z \wedge \, \mathcal{R}(a(\underline{z})b(\underline{w}))\,,
\end{equation}
where $S_{\underline{w}}^2$ is a sphere centred at $\underline{w}$ with fixed radius. To proceed, we notice that a sphere $B_1$ centred at the origin with radius $R_1 > r(\underline{w})$ can be homotopically deformed\footnote{Up to homotopy, $B_1$ may be written as the connected sum of $S^2_{\underline{w}}$ and $B_2$, and the contribution from the attaching cylinder can be shown to vanish.} into $S_{\underline{w}}^2$ and another sphere $B_2$, which is also centred at the origin but with radius $R_2< r(\underline{w})$; see Fig. \ref{fig:radialorder}. These spheres $S_{\underline{w}}^2,B_2$ of course have the same orientation as $B_1$.

\begin{figure}[h]
    \centering
    \includegraphics[width=0.7\linewidth]{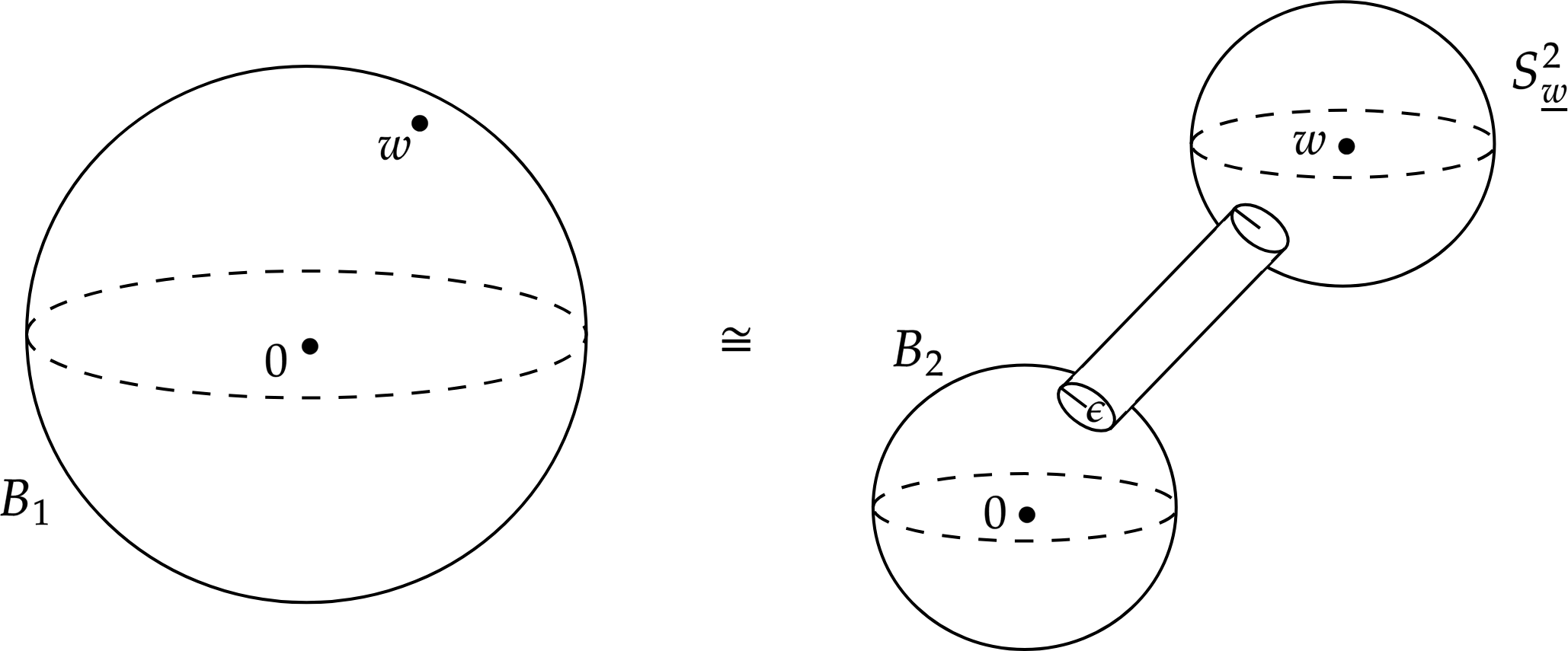}
    \caption{The configuration of spheres $S_{\underline{w}}^2,B_2$ upon deforming the sphere $B_1$. We then homotopically shrink the attaching cylinder of radius $\epsilon\to0$ to cancel its contributions from the commutator.}
    \label{fig:radialorder}
\end{figure}

As such, we can rearrange the integration over $S^2_{\underline{w}}$ into the integration over the larger sphere $B_1$ \textit{minus} that over the smaller sphere $B_2$.  Radial ordering then implies
\begin{equation}
\label{ec:cont2}
\begin{split}
      \oint_{S^2_{\underline w}} \d z \wedge \mathcal{R}(a(\underline z)b(\underline w))
&= \oint_{B_1} \d z \wedge a(\underline z)b(\underline w)-\oint_{B_2} \d z \wedge b(\underline w)a(\underline z)\\
&=[A,b(\underline w)] \,,
\end{split}
\end{equation}
where we have defined the operator
\begin{equation}
    A = \oint_B \d z \wedge a(\underline z) \,,
\end{equation}
with $B$ any fixed-radius sphere enclosing the insertion at $\underline w$. If we want this relation to hold as an operator identity, we must allow for an arbitrary number of additional insertions in a general correlator. In particular, the above deformation is justified only when $b(\underline w)$ is the unique insertion whose OPE with $a(\underline z)$ is singular in the region between the spheres $B_1$ and $B_2$.

Integrating this relation once more, we obtain the commutator of two operators defined by:
\begin{equation}
\label{ec:3dcommutator}
    [A,B]= \oint_{B_0} \d w \oint_{S^2_{\underline w}} \d z \wedge \mathcal{R}(a(\underline z)b(\underline w)) \,,
\end{equation}
where $B_0$ is a fixed-radius sphere enclosing the origin. Having established the commutator formula in the full three-dimensional geometry, we now return to the $\hat{\d}'$-cohomological representatives of the currents and again suppress the $\tau$-coordinate. Accordingly, in the computation of the mode algebra below we revert to the notation $z,w$. The resulting commutation relations are summarized in the following proposition.
\begin{prop}
\label{thm:higherkac}
    The commutation relations for the modes $\tL_n^a$, $\tH_m^a$ of the currents corresponding to the topological-holomorphic symmetry of the three-dimensional action \eqref{3daction} are given by 
\begin{equation}
\label{ec:higherpossbrack1}
      [\tH_n^a,\tH_m^b] = \tilde f^{ab}_c\tH_{n+m}^c \,, \quad [\tL^a_n,\tL^b_m] = 0 \,, 
\end{equation}
 \begin{equation}
 \label{ec:higherpossbrack2}
      [\tH_n^a,\tL_m^b]=\begin{cases}
          (\tilde \mu_2)^{ab}_c \tL_{m-n}^c + n(\kappa^{-1})^{ab} \delta_{n-1,m} \quad \text{if}\quad m\geq  n-1 \\
          0 \qquad\qquad \qquad\qquad\qquad\quad\qquad\text{otherwise}
      \end{cases} \,,
 \end{equation}
with the dual structure constants $\tilde f^{ab}_c$ and $(\tilde\mu_2)^{ab}_c$ defined in \eqref{dualbrackets}. 

\end{prop}

\begin{proof}
    We recall the mode expansions of the components of the currents in a basis of the Lie algebras $\fn$ and $V$ respectively
    \begin{equation}
        \tL^a(z) = \sum_{n=0}^\infty \tL^a_n z^n \,,\quad \tH^a(z) = \sum_{m=0}^\infty \tH^a_m \Omega^m \,,
    \end{equation}
where the modes can be written in terms of the components of the currents as 
\begin{equation}
    \tL_n^a = \frac{1}{8\pi i}\oint_{S^2} \dr z \wedge \Omega^n \tL^a(z) \,,\quad \tH^a_n=\frac{1}{8\pi i}\oint_{S^2} \dr z \wedge z^n  \tH^a(z)\,.
\end{equation}
Then, using the definition of the commutator \eqref{ec:3dcommutator} we can write
\begin{equation}
\begin{split}
    [\tH^a_n,\tH^b_m]
    &=\frac{1}{(8\pi i)^2}\oint_{S_0^2} \dr w \oint_{S_w^2} \dr z \wedge\,w^n \, z^m\, \mathcal{R}(\tH^a(z)\tH^b(w)) \\
    & =\frac{1}{(8\pi i)^2}\oint_{S_0^2} \dr w \oint_{S_w^2} \dr z \wedge\,w^n \, z^m  \left(\Omega_{z-w}^0 \tilde f^{ab}_d \tH^d(w) +\mathrm{reg}\right) 
\end{split}
\end{equation}
where in the second line we replaced the radially ordered product by the OPE \eqref{ec:HHOPE}, as the
singular part is invariant under radial ordering due to the antisymmetry of $\tilde f^{ab}_d$
and only this part contributes to the integral. Expanding $z^m$ around $w$ and using the generalised Cauchy formula \eqref{ec:gencauchyform} we find
\begin{equation}
    [\tH^a_n,\tH^b_m]=\frac{1}{8\pi i}\oint_{S_0^2} \dr w \wedge \,w^{n+m} \, \tilde f^{ab}_d \tH^d(w)= \tilde f^{ab}_d \tH_{n+m}\,.
\end{equation}
Next, we have that $[\tL_n^a,\tL_m^b]=0$ since the $\tL\tL$-OPE is regular \eqref{ec:LLOPE}. Finally, for the mixed commutator we have
\begin{equation}
\begin{split}
    [\tH_n^a,\tL_m^b] 
    &=\frac{1}{(8\pi i)^2}\oint_{S_0^2} \dr w \oint_{S_w^2} \dr z \wedge z^n \Omega_w^m \,\mathcal{R}(\tH^a(z)\tL^b(w)) \\
    &=\frac{1}{(8\pi i)^2}\oint_{S_0^2} \dr w \oint_{S_w^2} \dr z \wedge z^n\Omega_w^m\left(\Omega_{z-w}^0 (\tilde \mu_2)^{ab}_c \tL^c(w) + \Omega_{z-w}^1(\kappa^{-1})^{ab}\right) \\
    &= \frac{1}{(8\pi i)^2}\oint_{S_0^2} \dr w \,\Omega_w^m \oint_{S_w^2} \dr z \wedge \sum_{j=0}^n\binom{n}{j}(z-w)^j w^{n-j} \\
    &  \hspace{5cm}\times\left(\Omega_{z-w}^0 (\tilde \mu_2)^{ab}_c \tL^c(w) + \Omega_{z-w}^1(\kappa^{-1})^{ab} \right)\,,
\end{split}
\end{equation}
where in the second line we used the $\tH \tL$-OPE given in \eqref{ec:HLOPE} and in the third line we expanded $z^n$ around $w$. The first term contributes only for $j=0$, whereas the second term will be non vanishing for $j=1$. Thus, using the generalized cauchy formula \eqref{ec:gencauchyform} we obtain
\begin{equation}
\begin{split}
    [\tH_n^a,\tL_m^b] 
    &= \frac{1}{8\pi i}\oint_{S_0^2} \dr w \wedge\Omega_w^m \left(\,w^n\,(\tilde \mu_2)^{ab}_c \tL^c(w)  +n\,w^{n-1}(\kappa^{-1})^{ab} \right)\\
    &=(\tilde \mu_2)^{ab}_c \tL_{m-n}^c + n\,\delta_{n-1,m}(\kappa^{-1})^{ab} \,,
\end{split}
\end{equation}
where we used the relation $\Omega_w^m w^n = \Omega_w^{m-n}$ for $m\geq n$, and the generalized Cauchy formula \eqref{ec:gencauchyform} for the second term.

\end{proof}

\textbf{Theorem} \ref{thm:higherkac} gives us commutation relations of the modes of the currents associated to the infinite dimensional symemtry of the three-dimensional
action \eqref{3daction}. This is in complete analogy to the Kac-Moody algebra satisfied
by the modes of the conserved currents in the Wess-Zumino-Witten model. In
the next section, we will show that the current modes can be identified with the
generators of a centrally extended affine graded Lie algebra


\begin{remark}\label{ncchargealg}
    The Noether currents of the non-chiral 3d action \eqref{ec:nonchiralaction}, defined in terms of $\CB$,
    was studied in \S 6.3 of \cite{Chen:2023integrable}. The main result there was that the resulting differential graded commutator algebra --- analogues of \eqref{ec:higherpossbrack1}, \eqref{ec:higherpossbrack2} --- was found to be closely related to a higher derived notion of Lax integrability, similar to the situation with the Wess-Zumino-Witten model \cite{Alekseev1990-hq}. The authors expect that these methods could similarly be applied to analyze the Lax integrability of the chiral theory  \eqref{3daction}. 
\end{remark}

\section{The Centrally Extended Affine Graded Lie Algebra}
\label{sec:ceagla}

In the previous section, we computed the commutation relations of the currents associated with the infinite-dimensional symmetry of the three-dimensional action \eqref{3daction}. In this section, we prove that these commutation relations define
a centrally extended affine graded Lie algebra.

To begin, we consider the dual vector spaces \( \mathfrak{n}^\vee \) and \( V^\vee \), equipped with dual bases \( \{s^a\} \) and \( \{t^b\} \), respectively:
\begin{equation}
    \mathfrak{n}^\vee = \mathrm{span}(s^a) \,,\qquad V^\vee = \mathrm{span}(t^b)\,.
\end{equation}
These dual bases are constructed from the original bases \( \{t_a\} \) of \( \mathfrak{n} \) and \( \{s_b\} \) of \( V \), using the non-degenerate (invertible) bilinear form \( \langle \cdot, \cdot \rangle \), as follows:
\begin{equation}
    t^a = (\kappa^{-1})^{ab} \langle -, s_b \rangle \,,\qquad
    s^a = (\kappa^{-1})^{ab} \langle t_a, - \rangle\,.
\end{equation}
By construction, these satisfy the duality conditions:
\begin{equation}
    t^a(t_b) = \delta^a_b \,,\qquad s^a(s_b) = \delta^a_b\,.
\end{equation}
We endow the graded vector space \( \mathfrak{G}^\vee = \mathfrak{n}^\vee \oplus V^\vee \) with the following brackets:
\begin{equation}
    [t^a, t^b]^\vee = \tilde f^{ab}_c\, t^c \in V^\vee, \qquad
    \mu_2^\vee(t^a,s^b) = (\tilde\mu_2)^{ab}_c\, s^c \in \mathfrak{n}^\vee\,,
    \label{ec:aff2liealgbrack}
\end{equation}
where the structure constants are defined by
\begin{equation}
        \tilde f^{ab}_c = (\kappa^{-1})^{ax}(\kappa^{-1})^{by}f_{xy}^z\kappa_{zc},\qquad  (\tilde\mu_2)^{ab}_c = (\kappa^{-1})^{ax}\kappa_{cz}(\mu_2)_{xy}^z(\kappa^{-1})^{yb} \,,\label{ec:dualbrackets2}
    \end{equation}
as in equation~\eqref{dualbrackets}. In accordance with the literature  \cite{Chen:2012gz,Chen:2013,Bai_2013,chen:2022}, we will assign the cohomological degree $(-1)$ to $\fn^\vee$, and cohomological degree $0$ to $V^\vee$. The reason for this is given in \textit{Remark \ref{dualdegree}}.

\begin{remark}\label{dualdegree}
    Given a Lie 2-algebra $\fG = \fh\xrightarrow{\mu_1}\fg$, the dual Lie 2-algebra $\fG^\vee[1] = \fg^\vee\xrightarrow{\mu_1^\vee}\fh^\vee$ is given by the underlying dual complex together with a shift in degree \textit{down} by $1$, denoted by "$[1]$", such that $\fg^\vee$ has degree $(-1)$ and $\fh^\vee$ has degree $0$. This is to ensure that $\mu_1^\vee$ remains a cohomological differential which increases the degree by $1$, and hence $\fG^\vee[1]$ can still be treated as a Lie 2-algebra.
\end{remark} 

With the above grading conventions, we have the following:
\begin{prop}
The brackets defined in \eqref{ec:aff2liealgbrack} satisfy the graded Jacobi identity. Hence, they endow \(\fG^\vee =  \mathfrak{G}^\vee[-1] \) with the structure of a graded Lie algebra.
\end{prop}

\begin{proof}
We consider  
\begin{equation}
\begin{split}
[t^a,[t^b,t^c]^\vee]^\vee
&=\tilde f^{ad}_e \tilde f^{bc}_d t^e \\
&= (\kappa^{-1})^{ar} (\kappa^{-1})^{bs} (\kappa^{-1})^{ct} \big(f_{ry}^z f_{st}^y\big) \kappa_{ze} t^e\,,
\end{split}
\end{equation}
where we used the definition of the dual structure constants given in \eqref{ec:dualbrackets2}. Summing the three terms that appear in the graded Jacobi identity, we obtain 
\begin{multline}
    [t^a,[t^b,t^c]^\vee]^\vee + [t^c,[t^a,t^b]^\vee]^\vee + [t^b,[t^c,t^a]^\vee]^\vee = \\
    (\kappa^{-1})^{ar} (\kappa^{-1})^{bs} (\kappa^{-1})^{ct} \big(f_{ry}^z f_{st}^y+f_{ty}^z f_{rs}^y+f_{sy}^z f_{tr}^y\big) \kappa_{ze} t^e = 0
\end{multline}
The expression in parentheses vanishes by the Jacobi identity for the Lie bracket of $\fn$, whose structure constants are given by $f_{ab}^c$. Next, for the mixed brackets involving both $t^a$ and $s^b$, we compute:
\begin{equation}
    \begin{split}
         \mu_2^\vee(t^a,\mu_2^\vee(t^b,s^c))
    &=(\tilde\mu_2)^{ad}_e(\tilde\mu_2)^{bc}_d s^e\\
    &= \kappa_{ez}\big( (\mu_2)_{xy}^z(\mu_2)_{pq}^y\big) (\kappa^{-1})^{ax}(\kappa^{-1})^{bp}(\kappa^{-1})^{qc} s^e \,.
    \end{split}
\end{equation}
As in the previous case, this vanishes due to the fact that $\mu_2$ descends to a derivation on $V$, as shown in Proposition~\ref{prop:nisanalgebra}.
\end{proof}

\medskip

We denote by $\CK_u^\bullet$ the formal raviolo power series\footnote{Here, $\mathcal{K}$ is the version of the raviolo polynomials $\mathcal{K}_\text{poly}$ \eqref{ec:ravpoly} with $\mathbb{C}[u]$ in degree 0 replaced by the formal power series $\mathbb{C}[\![u]\!]$.} locally around the holomorphic variable $u=0$, such that its graded components are given by \( \CK_u^0 =\C[\![u]\!]\) and \( \CK_u^1 = \mathrm{span}_\C\{\Omega^0_u,\Omega^1_u,\dots\} \). With this at hand, we define the higher affine graded Lie algebra \( \widehat{\mathfrak{G}}^\vee \) as the total degree 0 piece of the tensor product complex
\begin{equation}
    \widehat{\mathfrak{G}}^\vee =\big( \fG^\vee\otimes\mathcal{K}^\bullet_u\big)_0 = \left(V^\vee \otimes \CK_u^0\right) \oplus \left(\fn^\vee \otimes \CK_u^1\right)\,.
\end{equation}
 The bracket on \( \widehat{\mathfrak{G}}^\vee \) is defined as follows:
For \( X \otimes f(u), Y \otimes g(u) \in V^\vee \otimes \CK_u^0 \), we set
\begin{equation}
    [X \otimes f(u), Y \otimes g(u)] = [X,Y]^\vee \otimes f(u)g(u) \in V^\vee \otimes \CK_u^0\,.
\end{equation}
For \( X \otimes f(u) \in V^\vee \otimes \CK_u^0 \) and \( A \otimes k(u) \in \fn^\vee \otimes \CK_u^1 \), the bracket is
\begin{equation}
    [X \otimes f(u), A \otimes k(u)] = \mu_2^\vee(X,A) \otimes f(u)k(u) \in \fn^\vee \otimes \CK_u^1\,.
\end{equation}
Finally, for two elements \( A \otimes k(u), B \otimes l(u) \in \fn^\vee \otimes \CK_u^1 \), the bracket vanishes:
\begin{equation}
    [A \otimes k(u), B \otimes l(u)] = 0\,.
\end{equation}
This bracket structure is consistent with the relations 
\begin{equation}
\label{ec:raviolorelations}
    u^n \Omega^m_u = 
    \begin{cases}
        0 & \text{if } n > m\,, \\
        \Omega^{m-n}_u & \text{if } n \leq m\,,
    \end{cases}
    \qquad \Omega^n_u \Omega^m_u = 0\,.
\end{equation}
In accordance with the above bracket structure, we shall assign a \textit{homological} grading to $\widehat{\fG}^\vee$ such that
\begin{equation}\label{degconvention}
     \text{deg}\left(V^\vee \otimes \CK_u^0\right) =0,\qquad \text{deg}\left(\fn^\vee \otimes \CK_u^1\right)=1\,.
\end{equation}
We shall adopt this degree convention in the following.

We can now show that the graded Lie brackets of $
\widehat{\fG}^\vee$ reproduce the commutators of the higher currents \eqref{ec:higherpossbrack1} and \eqref{ec:higherpossbrack2} up to the central term. Specifically, consider the bijective graded linear map
\begin{equation}
\label{ec:modesidentification}
    \tH_n^a \mapsto t^a \otimes u^n \,, \qquad \tL_m^a \mapsto s^a \otimes \Omega^m_u \,,
\end{equation}
which identifies the modes of the currents with the generators of the affine graded Lie algebra. Using the relations in equation~\eqref{ec:raviolorelations}, one readily verifies that the non-central contributions to the commutators in equations~\eqref{ec:higherpossbrack1} and~\eqref{ec:higherpossbrack2} coincide precisely with the graded Lie bracket on $\widehat{\mathfrak{G}}^\vee$.

\subsection{The Central Extension}

The final ingredient in the construction is the central extension of the affine graded Lie algebra $\widehat{\fG}^\vee$. This is determined by a
2-cocycle, whose definition requires an appropriate invariant bilinear form together with an analogue of the two-dimensional residue pairing.

For the residue pairing, we will use the integration along $S^2$ introduced in the previous sections which satisfies
\begin{equation}
\label{ec:res3}
    \frac{1}{8\pi i} \oint_{S^2} \mathrm{d}u \wedge u^n \Omega^m_u = \delta_{n,m}\,. 
\end{equation}
With this, we can define a residue pairing as a map 
\[
\operatorname{Res} =\operatorname{Res}_{u=0} : \mathcal{K}_u^\bullet \to \mathbb{C}[-1]
\]
which acts on the raviolo polynomial generators as
\begin{equation}
\label{ec:cuantodares}
    \operatorname{Res}(u^n) = 0\,, \quad \operatorname{Res}(u^n\Omega_u^m) = \frac{1}{8\pi i}\oint_{S^2}\mathrm{d}u\wedge u^n\Omega^m_u = \delta_{n,m}\,.
\end{equation}
The notation $\C[-1]$ reflects the fact that we are working with a chain complex with $\C$ in degree 1 --- namely, the notation "$[-1]$" indicates that the degree has been shifted \textit{up} by 1; cf. \textit{Remark \ref{dualdegree}}. Concretely, it indicates that $\operatorname{Res}$ vanishes on $(0,0)$-forms, while its action on $(1,0)$-forms yields values in $\C$, as explicitly shown in \eqref{ec:cuantodares}. 

On the other hand, to define a bilinear form on \( \mathfrak{G}^\vee = \mathfrak{n}^\vee \oplus V^\vee \), we make use of the bilinear pairing on \( H^\bullet(\mathfrak{G}) = V \oplus \mathfrak{n} \), which we establish to be non-degenerate and invariant in \textbf{Proposition}~\ref{prop:nisanalgebra}. Recall that, for a basis \( \{t_a\} \) of \( \mathfrak{n} \) and \( \{s_b\} \) of \( V \), this pairing is defined by
\(\langle t_a, s_b \rangle = \kappa_{ab}\) where \( \kappa_{ab} \) is an invertible matrix due to the non-degeneracy of the form. Taking the basis $\{s^a\}$ of $\fn^\vee$ and $\{t^b\}$ of $V^\vee$ we define the bilinear form $\langle\cdot,\cdot\rangle^\vee:V^\vee \times \fn^\vee \to \C$ by
\begin{equation}
\label{ec:kilform}
    \langle t^a,s^b\rangle^\vee = \big(\kappa^{-1}\big)^{ab} \,.
\end{equation}
In particular, we have the following
\begin{prop}
    The pairing $\langle \cdot,\cdot\rangle^\vee : V^\vee\times \fn^\vee \to \C$ is invariant. 
\end{prop}

\begin{proof}
We aim to show that the bilinear pairing \(\langle\cdot,\cdot\rangle^\vee\) on \(\fG^\vee\) is invariant under the dual Lie bracket structure, specifically that
\begin{equation}
    \langle [t^a,t^b]^\vee,s^c\rangle^\vee = \langle t^a, \mu_2^\vee(t^b,s^c)\rangle^\vee\,.
\end{equation}
Using the explicit expressions for the dual brackets from \eqref{ec:aff2liealgbrack} and the dual bilinear form from \eqref{ec:kilform}, this identity becomes
\begin{equation}
    (\kappa^{-1})^{dc}\,\tilde{f}^{ab}_d = (\kappa^{-1})^{ad}\,(\tilde{\mu}_2)^{bc}_d\,.
\end{equation}
Substituting the expressions for the dual structure constants in terms of the original ones as in \eqref{ec:dualbrackets2}, we find the condition
\begin{equation}
\label{ec:loquequieroprobar}
    (\kappa^{-1})^{ax}(\kappa^{-1})^{by} f_{xy}^c = (\kappa^{-1})^{bx} (\mu_2)^a_{xy} (\kappa^{-1})^{yc}\,.
\end{equation}

Now, recall that the invariance of the original bilinear pairing \(\langle\cdot,\cdot\rangle:\mathfrak{n} \times V \to \mathbb{C}\) implies
\begin{equation}
    \langle [t_a,t_b], s_c \rangle = \langle t_a, \mu_2(t_b, s_c) \rangle\,.
\end{equation}
Writing this in components we obtain
\begin{equation}
    f_{ab}^d\, \kappa_{dc} = \kappa_{ad} \, (\mu_2)_{bc}^d\,,
\end{equation}

which is precisely equivalent to \eqref{ec:loquequieroprobar} upon multiplying both sides by the inverse of \(\kappa\). This completes the proof.

\end{proof}

With both the residue pairing and the bilinear form in place, we are now ready to define the $2$-cocycle that determines the central extension of the affine graded Lie algebra. We consider, for $X \otimes f(u)\in V^\vee \otimes \CK_u^0$ and $A \otimes k(\Omega_u)\in \fn^\vee \otimes \CK_u^1$
\begin{equation}
\label{ec:2cocycledef}
    \mathfrak{w}(X \otimes f(u), A \otimes k(\Omega_u)) = \langle X, A \rangle^\vee \, \operatorname{Res}\left( \partial_u f(u) \, k(\Omega_u) \right),
\end{equation}
and extend it to be zero on pairs of elements both in \(V^\vee \otimes \CK_u^0\) and in \(\fn^\vee \otimes \CK_u^1\). 
\begin{prop}
    The bilinear map $\mathfrak{w} : \widehat{\fG}^\vee\times \widehat{\fG}^\vee \to \C[-1]$ determines a \textbf{graded Lie algebra 2-cocycle}. 
\end{prop}

\begin{proof}
To verify that $\mathfrak{w}$ defines a graded Lie algebra $2$-cocycle, we must check that
\begin{multline}
\mathfrak{w}\big([t^a \otimes u^n,\, t^b \otimes u^m]^\vee,\, s^c \otimes \Omega^k_u\big) 
+ \mathfrak{w}\big(\mu_2^\vee(s^c\otimes \Omega_u^k,\, t^a\otimes u^n),\, t^b\otimes u^m\big)\\
+ \mathfrak{w}\big(\mu_2^\vee(t^b\otimes u^m,\, s^c\otimes \Omega_u^k),\, t^a\otimes u^n\big) = 0\,.
\end{multline}
Using the structure constants and the identification $z^n \Omega^m = \Omega^{m-n}$, this becomes
\begin{multline}
\mathfrak{w}\big(\tilde{f}^{ab}_d\, t^d \otimes u^{n+m},\, s^c \otimes \Omega_u^k\big) 
+ \mathfrak{w}\big((\tilde{\mu}_2)^{ca}_d\, s^d \otimes \Omega_u^{k-n},\, t^b \otimes u^m\big) \\
+ \mathfrak{w}\big((\tilde{\mu}_2)^{bc}_d\, s^d \otimes \Omega^{k-m}_u,\, t^a \otimes u^n\big)\,.
\end{multline}
Applying the definition of $\mathfrak{w}$ and the identity $\partial_u \Omega^m_u = -(m+1)\Omega_u^{m+1}$, we find
\begin{multline}
\tilde{f}^{ab}_d\, (\kappa^{-1})^{dc} (n+m)\, \delta_{n+m-1,k}
- (\tilde{\mu}_2)^{ca}_d\, (\kappa^{-1})^{bd} (k-n+1)\, \delta_{m,k-n+1} \\
- (\tilde{\mu}_2)^{bc}_d\, (\kappa^{-1})^{da} (k-m+1)\, \delta_{n,k-m+1}\,.
\end{multline}
Next, we use the invariance of the bilinear form $\langle \cdot, \cdot \rangle^\vee$, which implies
\begin{equation}
\tilde{f}^{ab}_d\, (\kappa^{-1})^{dc}
= (\tilde{\mu}_2)^{bc}_d\, (\kappa^{-1})^{da}
= (\tilde{\mu}_2)^{ca}_d\, (\kappa^{-1})^{bd}\,.
\end{equation}
Substituting this into the previous expression, we obtain
\begin{multline}
\tilde{f}^{ab}_d\, (\kappa^{-1})^{dc}
\left[ (n+m)\, \delta_{n+m-1,k} - (k-n+1)\, \delta_{m,k-n+1} - (k-m+1)\, \delta_{n,k-m+1} \right] \\
= 2\, \tilde{f}^{ab}_d\, (\kappa^{-1})^{dc}\, (n+m-k-1)\, \delta_{n+m-k-1,0} = 0\,,
\end{multline}
which proves the 2-cocycle condition.
\end{proof}

Using the identification \eqref{ec:modesidentification}, it is straightforward to verify that this $2$-cocycle produces 
\begin{equation}
    \mathfrak{w} (t^a\otimes u^n,s^b\otimes \Omega_u^m) = \langle t^a,s^b\rangle^\vee \mathrm{Res} (nu^{n-1} \Omega_u^m) = n \big(\kappa^{-1}\big)^{ab}\delta_{n-1,m}\,,
\end{equation}
which is nothing but the central term in the commutator brackets of the currents given in \eqref{ec:higherpossbrack2}. We have thus shown:

\begin{theorem}\label{mainthm}
The algebra of the currents of the 3d theory \eqref{3daction} defines a centrally extended affine Lie algebra $\widehat{\fG}^\vee_{\mathfrak{w}}$ fitting into the central extension sequence
\begin{equation}
    0 \rightarrow \C[-1] \rightarrow \widehat{\fG}^\vee_{\mathfrak{w}} \rightarrow \widehat{\fG}^\vee \rightarrow 0 \,,\label{sesq}
\end{equation}
where the graded Lie bracket is given defined as follows:
For \( X \otimes f(u), Y \otimes g(u) \in V^\vee \otimes \CK_u^0 \), we have
\begin{equation}
    [X \otimes f(u), Y \otimes g(u)] = [X,Y]^\vee \otimes f(u)g(u) \,.
\end{equation}
For \( X \otimes f(u) \in V^\vee \otimes \CK_u^0 \) and \( A \otimes k(\Omega_u) \in \fn^\vee \otimes \CK_u^1 \), the bracket is
\begin{equation}
    [X \otimes f(u), A \otimes k(\Omega_u)] = \mu_2^\vee(X,A) \otimes f(u)k(\Omega_u) + \langle X,A\rangle^\vee\mathrm{Res}(\partial_u f(u)k(\Omega_u))\,.
\end{equation}
Finally, for two elements \( A \otimes k(\Omega_u), B \otimes l(\Omega_u) \in \fn^\vee \otimes \CK_u^1 \), the bracket vanishes:
\begin{equation}
    [A \otimes k(\Omega_u), B \otimes l(\Omega_u)] = 0\,.
\end{equation}

\end{theorem}

\begin{proof}
    The proof follows from the bijective identification 
    \begin{equation}
    \label{ec:identification}
        \tH_n^a \mapsto t^a\otimes u^n \,,\quad \tL_m^a \mapsto s^a\otimes \Omega^m_u\,.
    \end{equation}
\end{proof}

\subsection{Higher level quantization}

\textbf{Theorem \ref{mainthm}} naturally raises the question of whether the central extension of the affine graded Lie algebra $\widehat{\mathfrak G}^{\vee}_{\mathfrak w}$ is quantized. A useful point of reference is the ordinary affine Kac--Moody case, where the quantization of the level may be understood entirely in terms of the relation between central extensions of the loop algebra and central extensions of the loop group.

Indeed, for a complex simple Lie algebra $\mathfrak g$, the loop algebra $L\mathfrak g=\mathfrak g(\!(z)\!)$ admits, up to equivalence, a one-parameter family of central extensions
\begin{equation}
    0 \longrightarrow \mathbb C \longrightarrow \widehat{\mathfrak g}_{k} \longrightarrow L\mathfrak g \longrightarrow 0,
\end{equation}
since $H^2_{\mathrm{Lie}}(L\mathfrak g,\mathbb C)\cong \mathbb C$. After choosing a normalised generator of this one dimensional cohomology group, any central extension is determined by a scalar 
$k\in \C$, which is identified with the level of the extension. At the group level, however, one must further require that the corresponding Lie algebra extension integrate to a central extension of the loop group,
\begin{equation}
    1 \longrightarrow U(1) \longrightarrow \widehat{LG}_{k} \longrightarrow LG \longrightarrow 1.
\end{equation}
For $G$ compact, simple, and simply connected, such extensions are classified by the integral cohomology group
$H^3(G,\mathbb Z)\cong \mathbb Z$. Thus, although the loop algebra admits central extensions for arbitrary complex $k$, only those for which the corresponding class is integral integrate to genuine central extensions of the loop group. This is the origin of the quantization condition $k\in\mathbb Z$ in the ordinary affine setting.

An analogous argument may then be proposed for the graded extension sequence \eqref{sesq}. Provided the balancing $\kappa$ of the Lie $2$-algebra $\mathfrak G$ is unique up to scalar multiplication, one may expect the Lie $2$-algebra cohomology $H^2(\widehat{\mathfrak G}^{\vee},\mathbb C[-1])$ classifying the central extension, to be one-dimensional. In that case, after fixing a normalisation of the cocycle $\mathfrak w$ in \eqref{ec:2cocycledef}, the corresponding extension class is determined by a scalar $k\in\mathbb C$. The parameter $k$ may then be interpreted as the higher analogue of the affine level.
At the level of Lie $2$-groups, one expects a corresponding short exact sequence
\begin{equation}
\label{ec:groupextension}
    U(1)[-1]\to \widehat{\mathbb G}^{\vee}_{k}\to \widehat{\mathbb G}^{\vee},
\end{equation}
where the shifted group $U(1)[-1]=\mathbf B U(1)$ is understood as the delooping. Such an extension encods a multiplicative $U(1)$-bundle gerbe, ie. am appropriate higher group-cohomological class, on the infinite-dimensional Lie groupoid $\widehat{\mathbb G}^{\vee}$. The allowed values of the higher level $k$ should therefore be determined by the classification of such structures. This is currently part of the developing theory of smooth $2$-group extensions and bundle gerbes \cite{Ludewig:2023,Schommer_Pries_2011,Bunk:2020rju}, and the corresponding classification problem in the present infinite-dimensional setting does not yet appear to be known.

\section{The Raviolo Vertex Algebra}
\label{sec:vertexalgebra}

As discussed in the previous sections, much of the formalism developed there mirrors that of two-dimensional conformal field theories: the operator product expansion, radial quantisation, and the emergence of an infinite-dimensional symmetry algebra. A central feature in this setting, is the state-operator correspondence, which establishes a one-to-one relation between local operators in the field theory and states in the associated Hilbert space. It is this correspondence that ultimately allows one to pass consistently from the field-theoretic description to the Hilbert-space formulation.

In two dimensions, this formalism is rigorously captured by the notion of a vertex algebra. A vertex algebra provides a unified and mathematically precise framework that encompasses the various cases of interest—such as Virasoro and affine Kac–Moody algebras—while formalising many of the manipulations commonly used in practice. In particular, it offers a systematic way to organise and justify the procedures involved in the computation of physical quantities.

A key ingredient in \cite{Garner:2023zqn} is the introduction of the \textit{raviolo vertex algebra}, a three-dimensional analogue of the standard vertex algebra. From the perspective of quantum field theory, it encodes the algebra of local operators in a theory that is partially holomorphic and partially topological, just as ordinary vertex algebras capture the operator algebra of two-dimensional chiral theories.

In this section we show how the currents $\tL$ and $\tH$, together with their OPEs given in \textbf{Theorem} \ref{thm:opes}, can be used to construct a raviolo vertex algebra. For completeness, we begin by recalling the relevant definitions and results, and then proceed to the explicit construction of the raviolo vertex algebra associated with our three-dimensional theory \eqref{3daction}.

\subsection{A brief review of raviolo vertex algebras}

Before turning to the structures that define a raviolo vertex algebra, we introduce the \emph{raviolo delta function}, the analogue of the usual delta function on $\C$ but adapted to the geometry of $\C\times\R$. It is given by \cite{Garner:2023zqn}
\begin{equation}
\label{ec:deltafunction}
    \Delta(z-w) = \sum_{n\geq 0}w^n\Omega_z^n-\sum_{n\geq 0} z^n\Omega_w ^n\,, 
\end{equation}
and satisfies 
\begin{equation}
    \frac{1}{8\pi i}\oint_{S^2}\d z \wedge \Delta(z-w)f(z) = f(w)
\end{equation}
for any $f\in \mathcal{K}_{\mathrm{poly}}^\bullet$. Further properties of this distribution can be found in \cite{Garner:2023zqn}; we will only use its definition here.

With the delta distribution at hand, we consider a $\mathbb{Z}$-graded vector space $\mathcal{V}$, which will be the space of states of the theory. The grading is taken to be by cohomological degree. The essential object in this construction is the \textit{raviolo field} which will allow us to organize the action of a given local operator on $\CV$. 

\begin{dfn}
    A \textbf{raviolo field} on $\mathcal{V}$ is an element
    \begin{equation}
        A(z) = \sum_{m \geq  0} z^{m}A^+_m + \Omega_z^m A^-_m \,,
    \end{equation}
    such that for any $v\in V$ there exist $N$ sufficiently large with 
    \begin{equation}
        A^-_m v= 0\,\quad \text{for all }\,m\geq N\,.
    \end{equation}
We say that a field $A(z)$ is \textbf{homogeneous} of degree $|A|$ if, for every $m$, 
the mode 
$A^+_m : \CV \to \CV$ 
is an endomorphism of degree $|A|$, 
while $A^-_m : \CV \to \CV$
is an endomorphism of degree $|A|-1$. We denote the space of raviolo fields by $\mathcal{F}_{rav}(\CV)$.

\end{dfn}

In a QFT, one is often interested in quantities 
such as $\langle \Psi | A(z)^2 | \Psi \rangle$, the expectation value of the square of a quantum field $A(z)$ at a point $z$ in a state $|\Psi\rangle$. 
As seen explicitly in the OPE computations of the previous section, such expressions typically diverge. 
Physically, this reflects the fact that the fluctuations of $A(z)$ in the chosen state are infinite; 
equivalently, the variance of individual measurements of $A(z)$ is unbounded. 
What one ultimately seeks, however, is an operator whose classical limit is $A(z)^2$. 
To construct such an object, one must remove the short-distance divergences and retain a finite part, 
which can be interpreted as capturing the large-scale behaviour of the classical observable $A(z)^2$, 
rather than the short-scale fluctuations of the quantum field \cite{Hollands:2023txn}.  

In this setting, the operator product expansion provides a precise description of the singular behaviour 
that appears when fields are multiplied at coincident points. This in turn motivates the definition of the 
\emph{normal ordered product}, obtained by discarding the singular contributions.

\begin{dfn}
    Given $A(z)$ and $B(z)$ raviolo fields on $\mathcal{V}$, we define their \textbf{normal ordered product} $:A(z)B(w):$ as
    \begin{equation}
        :A(z)B(w): = A(z)_+B(w)+(-1)^{|A||B|}B(w)A(z)_-
    \end{equation}
where $A_+(z)=\sum_{m\geq 0}z^m A^+_m $ and $A_-(z) = \sum_{m\geq 0}A^-_m\Omega^m$    
\end{dfn}

The definition of the normal ordered product is entirely analogous to the two-dimensional case. It is given by the operator product expansion with the singular terms removed, that is, by the regular part of the OPE.

Normal ordering gives a consistent prescription for defining composite operators at coincident points. The other key structural ingredient is locality. In quantum field theory, locality means that the commutator of two fields vanishes when evaluated at distinct spacetime points, so that operators at spacelike separation commute. In the present setting, as seen in the OPEs computed in section \S \ref{affineraviolo}, the number of singular terms is always finite. This reflects the strong constraints imposed by the partial holomorphic–topological symmetry, which tightly restricts the possible short-distance behaviour of composite operators. It is precisely this rigidity that makes it possible to define a consistent algebraic structure governing the product of fields—something that would not be available in a more general quantum field theory. More precisely, we have

\begin{dfn}
    Given raviolo fields $A(z)$ and $B(w)$ on $\mathcal{V}$ we say that they are \textbf{mutually local} if there exist $N\geq 0$ such that 
    \begin{equation}
        (z-w)^{N+1}[A(z),B(w)]=0
    \end{equation}
    and 
    \begin{equation}
        \left(\Omega_z^m -\sum_{n=0}^N(w-z)^n\binom{m+n}{n}\Omega_w^{m+n}\right)[A(z),B(w)]=0 \,.
    \end{equation}
\end{dfn}
As shown in \cite{Garner:2023zqn}, mutual locality admits several equivalent formulations, closely mirroring those found in the theory of vertex algebras \cite{Kac1998,frenkel2004vertex}. The above definition makes it explicit that mutually local fields possess an OPE with only finitely many singular terms. In practice, however, it will be more convenient to use an equivalent formulation, stated in the following lemma.
\begin{lemma}
\label{ec:equivmutualloc}
    The raviolo fields $A(z)$ and $B(w)$ are mutually local if and only if there exist raviolo fields $C_n(w)$ such that 
    \begin{equation}
        [A(z),B(w)]=\sum_{n=0}^N \frac{1}{n!}\partial_w^n\Delta(z-w)C_n(w)
    \end{equation}
\end{lemma}
A proof of this lemma can be find in \textbf{Proposition 2.2.2} of \cite{Garner:2023zqn}. 

Together, normal ordering and mutual locality ensure that the collection of fields is closed under products: composite operators can be defined consistently, and their short-distance singularities are always under control. Beyond this, and in close analogy with the two-dimensional case, one can establish a one-to-one identification between raviolo fields and states in the vector space $\mathcal V$. Thus, in addition to possessing an algebra of operators, the theory is equipped with a vector space structure that organises these operators in a systematic way. This feature is known as the \emph{state–operator correspondence}, realised as a linear map from $\mathcal V$ to the space of raviolo fields on $\mathcal V$. That is,
\begin{equation}
    Y(-,z):\CV\to \cF_{rav}(\CV)\,,
\end{equation}
such that for every $a\in \CV$ we have
\begin{equation}
    Y(a,z)=\sum_{m \geq 0} z^{m} a^+_m + \sum_{m\geq 0} \Omega_z^m a^-_m \,. 
\end{equation}
With all of these definitions at hand, we can introduce the main definition of a \emph{raviolo vertex algebra} following \cite{Garner:2023zqn}

\begin{dfn}
    A \textbf{raviolo vertex algebra} is the data $(\mathcal{V},|0\rangle,\partial,Y)$ where 
    \begin{itemize}
        \item $\CV = \oplus \CV^r$ is the space of states, which is $\mathbb{Z}$-graded
        \item $|0\rangle$ is the vacuum vector which is a distinguished element in $\CV^0$
        \item $\partial:\CV \to \CV$ is the translator operator, given by a degree $0$ endomorphism
        \item $Y:\CV \to \mathcal{F}_{rav}(\CV)$ is a linear map of degree $0$
    \end{itemize}
subject to the following axioms
\begin{enumerate}
    \item For every $a\in \CV$ the element $Y(a,z)$ is a raviolo field
    \item For every $a\in \CV$ we have
    \begin{equation*}
        [\partial,Y(a,z)]=\partial_zY(a,z)
    \end{equation*}
    \item The vacuum vector satisfies $\partial|0\rangle = 0$, $Y(|0,\rangle,z)=\mathbb{1}_{\CV}$, and for every $a\in \CV$ we have $Y(a,z)|0\rangle \in \mathcal{V}[\![z]\!]$; that is, $Y_-(a,z)|0\rangle = 0$. 
    \item For every $a,b \in \CV$ the fields $Y(a,z)$ and $Y(b,w)$ are mutually local. 
\end{enumerate}

\end{dfn}

The definition of a raviolo vertex algebra involves the graded vector space $\CV$, the state–operator correspondence, and the mutual locality of fields. These elements are sufficient to define the fields of the theory and to prescribe how they compose through their operator products. To capture the full physical structure of the theory, however, two further ingredients are required. The first is the vacuum state, representing the ground state of the theory, whose axioms combine the expected physical features of the vacuum with the conditions needed for mathematical consistency. The second is the translation operator $\partial$, which generates infinitesimal displacements along the complex direction. It allows one to translate the notion of differentiating fields into an algebraic operation on states. 

With these ingredients in place, one obtains a complete algebraic description of the local operators 
of a three-dimensional theory with the symmetry described in the previous sections. 
Much as in the two-dimensional case of vertex algebras, this structure makes it possible to relate 
representation theory to physical quantities, thereby extending many of the familiar applications 
of chiral vertex algebras to the raviolo setting.

In practice, the existence of a raviolo vertex algebra is most easily established by means of a reconstruction theorem, in direct analogy with the two-dimensional case \cite{Frenkel:1994em}. Such theorems provide an alternative set of conditions which guarantee the existence of a raviolo vertex algebra structure, but are often simpler to verify in concrete situations. More precisely, following \cite{Garner:2023zqn} we have

\begin{prop}
\label{prop:rectheorem}
Let $\mathcal V = \bigoplus \mathcal V^r$ be a $\mathbb Z$--graded vector space, 
$|0\rangle$ a non--zero vector, $\partial$ a degree $0$ endomorphism of $\mathcal V$. 
Further, let $\{a^i\}$ be a countable ordered set of vectors in $\mathcal V$, 
with $a^i \in \mathcal V^{r_i}$, and suppose we are given homogeneous fields
\begin{equation}
    A^i(z) = \sum_{m\geq 0} z^{m} A^{i,+}_m +  \Omega^m_z A^{i,-}_m
\end{equation}
of degree $r_i$, such that the following hold:
\begin{enumerate}
    \item $A^{i,+}_{0}|0\rangle = a^i$ and $A^{i,-}_m |0\rangle = 0$ for all $i$ and $m \geq 0$.
    \item $\partial |0\rangle = 0$ and $[\partial, A^i(z)] = \partial_z A^i(z)$ for all $i$.
    \item All fields $A^i(z)$ are mutually local.
    \item $\mathcal V$ is spanned by the vectors
    \begin{equation}
        A^{i_1}_{j_1} \cdots A^{i_l}_{j_l} |0\rangle, \qquad j_k \geq 0\,.
    \end{equation}
\end{enumerate}
Then the assignment
\begin{multline}
   Y\!\left(A^{i_1}_{j_1}\cdots A^{i_l}_{j_l}|0\rangle, z\right) \\
= \frac{1}{(j_1-1)! \cdots (j_l-1)!} 
:\partial_z^{-j_1-1} A^{i_1}(z) \cdots \partial_z^{-j_l-1} A^{i_l}(z): 
\end{multline}
determines a well-defined raviolo vertex algebra structure on $\mathcal V$. 
Moreover, this is the unique raviolo vertex algebra structure on $\mathcal V$ 
satisfying conditions (1)--(4) and such that $Y(a^i, z) = A^i(z)$.
\end{prop}

\subsection{Raviolo Current Vertex Algebra}

In this final section, we construct the raviolo current vertex algebra associated with the anti-chiral sector of the topological-holomorphic theory \eqref{3daction}. 
Our strategy is to leverage the canonical representation of $\widehat{\fG}_{\mathfrak{w}}^\vee$ on its universal envelope, then apply the reconstruction theorem (Prop. 4.0.1 of \cite{Garner:2023zqn}) introduced above. 

Let us begin by constructing the vacuum module $\mathcal{V}$. From a physicist’s perspective, the vacuum module is simply the Fock space — the vector space generated by arbitrary products of creation operators acting on the vacuum. Since our raviolo vertex algebra will be built from the affine graded Lie algebra $\widehat{\mathfrak{G}}_{\mathfrak{w}}^\vee$ constructed in \S \ref{sec:ceagla}, we must identify, among its generators, which act as creation and which as annihilation operators.

From point (1) of Proposition~\ref{prop:rectheorem}, the creation operators are the $A^{i,+}_n$, which in the raviolo field expansion come proportional to powers of $z$. The annihilation operators, on the other hand, are the $A^{i,-}_n$ which come proportional to the form-valued basis elements $\Omega_z^n$. Given the raviolo fields are homogeneous in cohomological degree, their components must pair consistently with $z$ and $\Omega_z$. This determines that the generators $\tilde{L}_n^a$ are paired with $z$ and thus act as creation operators, while the generators $\tilde{H}_n^a$ are paired with $\Omega_z$ and act as annihilation operators.

The vacuum module $\mathcal{V}$ is therefore spanned by states of the form  
\begin{equation}
    \tilde{L}_{j_1}^{a^1} \tilde{L}_{j_2}^{a_2} \cdots \tilde{L}_{j_k}^{a_k} |0\rangle\,, 
    \qquad j_i \geq 0,
\end{equation}
and such that  
\begin{equation}
    \tilde{H}_n^a |0\rangle = 0\,, 
    \qquad n \ge 0\,.
\end{equation}

To formalize this intuitive picture, one can describe the construction of $\mathcal{V}$ in purely algebraic terms.  
Consider the affine graded Lie algebra $\widehat{\mathfrak{G}}^\vee_{\mathfrak{w}}$, with the grading conventions \eqref{degconvention}, generated by the homogeneous basis elements $\ell^a_n = s^a \otimes \Omega_u^n $ in degree 1 and $\eta^a_n= t^a\otimes u^n$ in degree 0. Its \textit{universal enveloping graded algebra}, denoted $U\widehat{\mathfrak{G}}_{\mathfrak{w}}^\vee$, is defined as the free graded algebra generated by polynomials in $\ell_n^a$ and $\eta_n^a$, together with a central element $\mathfrak{w}$ of degree~1.  

The generators are ordered according to the graded Poincaré–Birkhoff–Witt theorem, and obey the affine Lie algebra relations 
\begin{gather*}
    \ell^a_n\ell^b_m - \ell^b_m\ell^a_n = 0\,, 
    \qquad 
    \eta^a_n\eta^b_m - \eta^b_m\eta^a_n = [\eta^a_n,\eta^b_m]\,,\\[4pt]
    \eta^a_n\ell^b_m - \ell^b_m\eta^a_n = [\eta^a_n,\ell^b_m]\,.
\end{gather*}
The homogeneous degree-0 unit is $(\ell^a_n)^0 = (\eta^a_m)^0 = 1$. By working under the identification \eqref{ec:identification},
\begin{equation}
\label{ec:identification2}
    \tL_n^a \mapsto \ell^a_n = s^a \otimes \Omega_u^n \,,\quad \tH_n^a \mapsto\eta^a_n= t^a\otimes u^n
\end{equation}
the mode generators $\tL^a_n,\tH^a_n$ act canonically on $U\widehat{\fG}_\mathfrak{w}^\vee$ by multiplication by $\ell^a_n,\eta^a_n$, respectively. We will use \eqref{ec:identification2} extensively in the following.

Next, consider the \textit{central factor} $\mathbb{C}_w = \mathbb{C} \oplus \mathbb{C}w[-1]$, where $w$ is an abstract Abelian generator of degree~1, and whose homogeneous degree-0 unit is $1 \oplus 0w = 1$.  
Taking the tensor product $U\widehat{\mathfrak{G}}^\vee_{\mathfrak{w}} \otimes \mathbb{C}_w$, the vacuum module $\mathcal{V}$ is obtained by quotienting out two relations:  
(i) all of the $\tH$-modes act trivially on $\mathcal{V}$, and  
(ii) the generator of the central piece $\mathbb{C}[-1] \subset \widehat{\mathfrak{G}}^\vee_{\mathfrak{w}}$ acts as multiplication by the constant $w$.\footnote{Note that both the values of the 2-cocycle $\mathfrak{w}$ and the scalar $w$ have cohomological degree~1.}  
Equivalently,
\begin{equation}
\label{ec:vacummodule}
   \mathcal{V} = U\widehat{\mathfrak{G}}^\vee_{\mathfrak{w}} \otimes_{U\widehat{\mathfrak{G}}^\vee_{\geq 0}} \mathbb{C}_w\,, 
\end{equation}
where $\widehat{\mathfrak{G}}^\vee_{\geq 0} = \mathrm{span}\{\tH_n^a\} \oplus \mathbb{C}w[-1]$ denotes the positive graded subalgebra.

The following equivalent mathematical definition should be familiar to experts in vertex algebras.
\begin{dfn}
    The \textbf{vacuum module} $\mathcal{V}=\operatorname{Ind}_{U\widehat\fG^\vee_{\geq 0}}^{U\widehat{\fG}^\vee_\mathfrak{w}}\C_w$ is the induced representation of the trivial $\widehat \fG^\vee_\mathfrak{w}$-representation on $\C_w\cong \C\oplus\C[-1]$.
\end{dfn}
\noindent In the following, we shall leverage \eqref{ec:identification2} in order to distinguish a \textit{vector} $\ell^a_n,\eta^a_n\in\widehat{\fG}^\vee_\mathfrak{w}\subset \mathcal{V}$, and its corresponding \textit{operator} $\tL^a_n,\tH^a_n$ on $\mathcal{V}$.


\begin{theorem}
    The conserved currents arising from the three-dimensional action \eqref{3daction} give rise to a raviolo vertex algebra structure.
\end{theorem}

\begin{proof}
Let us take the vector space $\mathcal{V}$ defined in \eqref{ec:vacummodule} and consider, for $1\leq a\leq \operatorname{dim}\mathfrak{n}^\vee=\operatorname{dim}V^\vee$, the raviolo field operators
\begin{equation}
    \mathcal{J}^a(z) = \sum_{n\geq 0} z^n \tL_n^a  + \Omega_z^n \tH^a_n
\end{equation}
which, with the degree conventions in \textit{Remark \ref{dualdegree}}, are homogeneous of degree $1$. We identify the unit $ |0\rangle=1\otimes 1\in\mathcal{V}$ as the vacuum. We will now prove points (1) through (4) of \textbf{Proposition} \ref{prop:rectheorem}.

By definition of $\mathcal{V}$, we have that
\begin{equation}
    \tH_n^a |0\rangle=0 \,,\quad  \tL_0^a|0\rangle = \ell_0^a\,,\label{vacmodulepolynomials}
\end{equation}
which proves point (1). Next, for the translation operator, we 
begin by defining an additional generator $D$ on $\widehat{\fG}^\vee_\mathfrak{w}\cong (\fG^\vee\otimes\mathcal{K}^\bullet_u)\oplus \C[-1]$ from the action of the derivative $\partial_u$ on the current modes:
\begin{equation*}
    [D,X\otimes k(u)] = X\otimes\partial_u k(u), \qquad [D,\mathfrak{w}]=0\,,
\end{equation*}
where $X\in\fG^\vee,~ k(u)\in\mathcal{K}_u^\bullet$ and $\mathfrak{w}\in\C[-1]$ is the degree 1 generator.  Under the canonical $\widehat{\fG}^\vee_\mathfrak{w}$-representation  on $\mathcal{V}$, this gives rise to an endomorphism $\partial\in\operatorname{End}\mathcal{V}$ which, by construction, satisfies
\begin{equation}
    [\partial,x]=\partial_z x\,,\qquad \partial|0\rangle = 0\label{translate}
\end{equation}
for all $x\in \widehat{\fG}^\vee_\frak{w}$. More explicitly,  these read as the following mode relations
\begin{equation}
\label{ec:actionpartial1}
    [\partial ,\tL_n^a] = (n+1)\tL_{n+1}^a\,,\quad [\partial, \tH_n^a] = -n \tH_{n-1}^a\,.
\end{equation}
This then allows us to compute
\begin{equation}
\label{ec:actionpartial2}
\begin{split}
    [\partial,\mathcal{J}^a(z)]
    &=\sum_{n\geq 0}z^n[\partial,\tL_n^a]+ \Omega_z^n[\partial,\tH_n^a] \\
&\stackrel{\eqref{ec:actionpartial1}}{=} \sum_{n\geq 0 }(n+1)z^n \tL_{n+1}^a - n\Omega^{n}_z\tH_{n-1}^a=\partial_z \mathcal{J}^a(z)\,,
\end{split} 
\end{equation}
identifying $\partial$ as the desired translation operator, thereby proving point (2). 

\medskip

To prove point (3), mutual locality, we will directly compute the commutator $[\mathcal{J}^a(z),\mathcal{J}^b(w)]$ between two raviolo fields and show that it is of the form in \textbf{Lemma \ref{ec:equivmutualloc}}. That is, we will show that 
\begin{equation}
\label{ec:point3aim}
    [\mathcal{J}^a(z),\mathcal{J}^b(w)]= (\bm{\kappa}^{-1})^{ab}\partial_w\Delta(z-w) + \mathbf{\tilde f}^{ab}_c\Delta(z-w)\mathcal{J}^c(w)\,.
\end{equation}
This computation involves a few grading-related technicalities, so we will present it in full detail.

To begin with, although we have been using the same Lie algebra index for both $\tL$  and $\tH$, in fact these fields take values in two distinct vector spaces. It will therefore be important to distinguish them in what follows. We thus introduce a collective index $a=(A,\alpha)$ so that
\begin{equation}
    \mathcal{J}^{a}(z)=\sum_{n\geq 0}z^n\tL^A_n + \Omega_z^n\tH^\alpha_n \,.
\end{equation}
Furthermore, recall from \eqref{degconvention} that $\tilde L_n^A$ has degree $1$ while $\tilde H_n^\alpha$ has degree $0$, and that $z^n$ has degree zero and $\Omega_z^n$ has degree $1$. As such, the raviolo field $\mathcal{J}^a(z)$ has homogeneous degree $1$ in $\mathcal{K}^\bullet_z \otimes \mathfrak{G}^\vee$. The commutator can thus be computed as
\begin{multline}
\label{ec:finalcom}
    [\CJ^a(a),\CJ^b(w)]=\sum_{n,m\geq 0} z^nw^m[\tL_n^A,\tL_m^B]-z^n\Omega^m_w[\tL_n^A,\tH_m^\beta]\\
    +\Omega_z^nw^m[\tH_n^\alpha,\tL_m^B]+\Omega_z^n\Omega_w^m[\tH_n^\alpha,H_m^\beta]\,,
\end{multline}
where a minus sign appeared in the second term from commuting $\tilde L_n^A$ past $\Omega_w^m$, both of which have odd degree $1$.\footnote{We thank Niklas Garner for pointing this out.}

At this point we must be careful when replacing the mixed commutators in \eqref{ec:finalcom} by their explicit expressions, which we recall here for convenience:
\begin{equation}
\label{ec:higherpossbrack4}
[\tH_n^\alpha,\tL_m^B]
=
\begin{cases}
(\tilde \mu_2)^{\alpha B}_C\,\tL_{m-n}^C
+ n(\kappa^{-1})^{\alpha B}\,\delta_{n-1,m}
& \text{if } m \ge n-1 \,,\\[0.3em]
0 & \text{otherwise} \,.
\end{cases}
\end{equation}
Recall that $(\tilde\mu_2)^{\alpha B}{}_C$ and $(\kappa^{-1})^{\alpha B}$ were defined in
\eqref{ec:aff2liealgbrack} and \eqref{ec:kilform}, respectively, as the structure constants of the dual maps
\begin{equation}
  \mu_2^\vee : V^\vee \times \fn^\vee \to \fn^\vee,
  \qquad
  \langle\,\cdot\,,\,\cdot\,\rangle^\vee : V^\vee \times \fn^\vee \to \C .
\end{equation}
Hence, in order to write the commutator
$[\tL_n^A,\tH_m^\beta]$ appearing in \eqref{ec:finalcom}, we introduce
the map
\begin{equation}
\nu_2^\vee : \fn^\vee \times V^\vee \to \fn^\vee \,,
\qquad
\nu_2^\vee(s^A,t^\beta)
= -\,\mu_2^\vee(t^\beta,s^A)\,,
\end{equation}
so that, in components,
\[
(\tilde \nu_2)^{A\beta}_C
= -(\tilde \mu_2)^{\beta A}_C \,.
\]
In complete analogy, we define
\begin{equation}
\eta^{-1} : \fn^\vee \times V^\vee \to \mathbb{C} \,,
\qquad
\eta^{-1}(s^A,t^\beta)
= \kappa^{-1}(t^\beta,s^A)\,,
\end{equation}
such that
\[
(\eta^{-1})^{A\beta}
= (\kappa^{-1})^{\beta A} \,.
\]
These definitions are chosen precisely so as to ensure the skew-symmetry of the bracket $[\tL_n^A,\tH_m^\beta]=-[\tH_m^\beta,\tL_n^A]$. With these conventions in place, we may rewrite \eqref{ec:finalcom} as
\begin{multline}
\label{ec:expparacomp}
        [\CJ^a(z),\CJ^b(w)]=\sum_{n,m\geq 0}-z^n\Omega_w^m\left[(\tilde{\nu}_2)^{A\beta}_C\tL_{n-m}^C + m(\eta^{-1})^{A\beta }\delta_{m-1,n}\right]\\
        +\Omega_z^nw^m\left[(\tilde{\mu}_2)^{\alpha B}_C\tL_{m-n}^C + n(\kappa^{-1})^{\alpha B}\delta_{n-1,m}\right] + \Omega_z^n\Omega_w^m \tilde{f}^{\alpha\beta}_\gamma \tH_{n+m}^\gamma\,.
\end{multline}

\medskip

Turning now to the right-hand side of equation \eqref{ec:point3aim}, we compute 
\begin{equation}
    \Delta(z-w)\mathcal{J}^c(w)=\sum_{n,m\geq 0}w^{n+m}\Omega_z^n\tL_m^C + \Omega_z^n\Omega_{w}^{m-n}\tH_m^\gamma -z^n\Omega_w^{n-m}\tL_m^C
\end{equation}
which is an element of 
\begin{equation}
\label{ec:dirsumdec1}
    (\mathcal{K}_w^0\otimes \mathcal{K}_z^1\otimes \fn^\vee)\oplus(\mathcal{K}_w^1\otimes \mathcal{K}_z^1\otimes V^\vee)\oplus(\mathcal{K}_w^1\otimes \mathcal{K}_z^0\otimes \fn^\vee)\,.
\end{equation}
We define the linear map $\mathbf{\tilde f}^{ab}_c$ acting on this space as the endomorphism:
\begin{equation}
    \mathbf{\tilde f}^{ab}_c=\iota_1 \circ (\tilde{\mu}_2)^{\alpha B}_C \circ p_1 + \iota_2 \circ \tilde f^{\alpha \beta}_\gamma\circ p_2 + \iota_3 \circ (\tilde{\nu}_2)^{A\beta}_C \circ p_3
\end{equation}
where $\iota_i,p_i$ with $i=1,2,3$ are inclusions and projections, respectively, of the direct sum decomposition \eqref{ec:dirsumdec1}. By construction, we then have 
\begin{multline}
\label{ec:granf1}
    \mathbf{\tilde f}^{ab}_c\Delta(z-w)\mathcal{J}^c(w) = \sum_{n,m\geq 0}w^{n+m}\Omega_z^n(\tilde{\mu}_2)^{\alpha B}_C\tL_m^C + \Omega_z^n\Omega_{w}^{m-n}f^{\alpha \beta}_\gamma\tH_m^\gamma \\
    -z^n\Omega_w^{n-m}(\tilde{\nu}_2)^{A\beta}_C\tL_m^C\,.
\end{multline}
In complete analogy, we have that
\begin{equation}
    \partial_w\Delta(z-w)=\sum_{n\geq 0} nw^{n-1}\Omega_z^n + (n+1)z^n\Omega_w^{n+1}\,,
\end{equation}
which can be thought as an element of 
\begin{equation}
\label{ec:dirsumdec2}
    (\CK_w^0\otimes \CK_z^1)\oplus(\CK_w^1\otimes \CK_z^0)\,,
\end{equation}
so that we may define the contraction with $(\bm{\kappa}^{-1})^{ab}$ as a linear endomorphism of the above space given by 
\begin{equation}
    (\bm{\kappa}^{-1})^{ab}=\iota_1\circ (\kappa^{-1})^{\alpha B}\circ p_1 - \iota_2\circ (\eta^{-1})^{A\beta}\circ p_2\,,
\end{equation}
where again, $\iota_j,p_j$ with $j=1,2$ are inclusions and projections, respectively, of the direct sum decomposition \eqref{ec:dirsumdec2}. By definition, we thus obtain
\begin{equation}
\label{ec:grank1}
    (\bm{\kappa}^{-1})^{ab}\partial_w\Delta(z-w)=\sum_{n\geq 0} n(\kappa^{-1})^{\alpha B}w^{n-1}\Omega_z^n - (n+1)(\eta^{-1})^{A\beta}z^n\Omega_w^{n+1}\,.
\end{equation}
Putting together equations \eqref{ec:granf1} and \eqref{ec:grank1}, and comparing with \eqref{ec:expparacomp} we finally have 
\begin{equation}
    [\mathcal{J}^a(z),\mathcal{J}^b(w)]= (\bm{\kappa}^{-1})^{ab}\partial_w\Delta(z-w) + \mathbf{\tilde f}^{ab}_c\Delta(z-w)\mathcal{J}^c(w)\,
    \label{commutator}
\end{equation}
as desired.

\medskip

Finally, point (4) follows immediately by construction of $\mathcal{V}$, and thus by the reconstruction \textbf{Theorem \ref{prop:rectheorem}}, $\mathcal{V}$ defines a unique raviolo vertex algebra structure.

\end{proof}

\begin{remark}
    The expression of \eqref{commutator} bears a striking resemblance to the commutator bracket obtained for the "raviolo current algebra" in \cite{Garner:2023zqn}. However, one crucial difference is that our pairing form $\kappa^{-1}$ has odd degree, which leads to a completely different presentation \eqref{vacmodulepolynomials} of the raviolo vacuum module $\mathcal{V}$. Moreover, this non-trivial degree within $\kappa$ is in fact closely related to the \textit{2-graded classical Yang-Baxter equations} \cite{Bai_2013,chen:2022,Cirio:2012be}, as well as the construction of higher-braiding structures \cite{Chen:2023tjf,Kemp_2025} --- both aspects relating back to higher-integrability (cf. \cite{Chen:2023integrable}).
\end{remark}

\medskip

\section*{Future Directions}

As mentioned in the introduction, a central motivation for studying infinite-dimensional symmetry algebras is that, in favourable situations, they can impose strong constraints on the dynamics of a theory. Having identified such an infinite-dimensional symmetry structure for the cohomology classes of the currents in our model, it is then important to understand whether it can play a comparable role in the present setting..

One standard way in which this happens in the two-dimensional case is through the representation theory of the vertex algebra. More precisely, the states of the theory are organized into modules of the chiral algebra, which provides a natural framework for analysing operators and their descendants by reducing the problem to the study of distinguished primary states from which larger families of operators are generated by the action of the symmetry modes. It is thus natural to ask whether a similar picture may hold in the present setting. In this direction, \cite{Garner:2023zqn} initiates the study of modules for raviolo vertex algebras and shows, in the affine case, that these are equivalent to modules of the associated current Lie algebra. In the present case, where the symmetry currents give rise to a centrally extended affine graded Lie algebra, one may similarly expect that its representations provide a natural organizing principle for the operators described by the cohomology of the currents. We leave a detailed study of this question for future work.

Another way in which infinite-dimensional symmetry algebras constrain two-dimensional theories is through conformal blocks. For operator insertions at marked points \(z_1,\dots,z_N\), conformal blocks are the chiral quantities attached to that configuration. As one varies the marked points, the conformal blocks vary accordingly, and in the affine case this dependence is constrained by the Ward identities of the current algebra, which give rise to the Knizhnik--Zamolodchikov equations governing their dependence on \(z_1,\dots,z_N\). In this direction, \cite{Alfonsi:2024qdr} introduces a raviolo configuration space adapted to the topological-holomorphic setting and uses it to define the corresponding spaces of coinvariants, whose duals provide the raviolo analogue of conformal blocks. They further show that the raviolo state-field map is recovered from the collision limits of these objects. This suggests that the current algebra constructed here may also act on such raviolo conformal blocks, and that the corresponding Ward identities could impose additional constraints on them, perhaps leading to an analogue of the Knizhnik–Zamolodchikov equations. It would be very interesting to explore this further.

Finally, we would like to observe that the analysis carried out here only captures the operator products at the quadratic order. In general, we expect the full quantum structure to also include additional higher-order products, arising from homotopy transfer;\footnote{Contrast this with the special case $\mu_1=0$; here, the Lie 2-algebra $\fG=H^\bullet(\fG)$ is its own cohomology and hence homotopy transfer need not be performed. However, this case corresponds to a cohomological vertex algebra, which has no non-trivial higher-order operator products \cite{Griffin:2025add}.} see Appendix \ref{sec:appendixproof}. Such techniques of extracting higher operator products in dimensions greater than two have been studied extensively, for example in the factorization-algebra framework of \cite{costello_gwilliam_2016}, in the operatope approach of \cite{Budzik:2022mpd,Gaiotto:2024gii}, and in recent models for holomorphic and topological-holomorphic field theories \cite{Felder:2025bsz,Costello:2020ndc,Zeng:2023qqp,Zeng:2021zef}. In the same spirit, the full higher vertex-algebraic quantization of 
the theory, and its relation to quantum higher integrability/Zamolodchikov's tetrahedron equations, will be studied elsewhere.


\subsection*{Acknowledgements}

Part of the completion of this project was supported by the China international talents Exchange Program (CEP program code JC202502007G) at the Beijing Institute of Mathematical Sciences and Applications. JL thanks Lewis Cole for valuable comments on the introduction, and Horacio Falomir for his support as PhD advisor. The work of JL is supported by CONICET. The work of HC is supported by the National Science Foundation of China (grant number W2533012).

\appendix

\section{Lie 2-Groups and Lie 2-Algebras}

\label{sec:appLie2alg}
Here we provide basic definitions of Lie 2-groups and Lie 2-algebras. They are part of the higher homotopy $L_\infty$-algebras generalizing Lie algebras \cite{Kim:2019owc,Baez:2004,Bai_2013}. 

\begin{dfn}
    A \textbf{Lie 2-group} $\mathbb{G} \rightrightarrows G$ is a Lie groupoid over a smooth manifold $G$, equipped with a smooth product functor $\otimes: \mathbb{G}\times \mathbb{G}\to \mathbb{G}$, such that all objects in $G$ are $\otimes$-invertible. 
\end{dfn}
We say that the Lie 2-group $\mathbb{G}$ is \textbf{strict} iff its invertible monoidal associator $\alpha: \otimes \circ (\otimes \times 1) \Rightarrow \otimes\circ (1\times\otimes)$ is the identity. In which case, its base $G$ is a Lie group.

There is an equivalent but convenient characterization for strict Lie 2-groups in terms of a pair of ordinary Lie groups. 
\begin{dfn}
\label{2grpdef}
A {\bf Lie group crossed-module} $\mathbb{G}=(\mathsf{H}\xrightarrow{\bar\mu_1}G,\rhd)$ consists of two (Lie) groups $\mathsf{H},G$, a (Lie) group homomorphism $\bar\mu_1:\mathsf{H}\rightarrow G$ and an (smooth) action $\rhd:G\rightarrow \operatorname{Aut}\mathsf{H}$ such that the following conditions
\begin{equation}
    \bar\mu_1(x\rhd y) = x\bar\mu_1(y)x^{-1}\,,\qquad (\bar\mu_1(y))\rhd y'=yy'y^{-1}\label{pfeif2}
\end{equation}
are satisfied for each $x\in G$ and $y,y'\in \mathsf{H}$.
\end{dfn}
\noindent Indeed, the category of strict Lie 2-groups is equivalent to the category of Lie group crossed-modules \cite{Porst2008Strict2A}.

\medskip

The infinitesimal approximation of a Lie 2-group gives rise to a Lie 2-algebra, which is a strict 2-term $L_\infty$-algebra \cite{Baez:2003fs}. 
\begin{dfn}\label{lie2alg}
Let $\mathbb{K}$ denote a field of characteristic zero (such as $\R$ or $\C$). A {\bf Lie 2-algebra} $\fG=(\fh\xrightarrow{\mu_1}\fg,\mu_2)$ over $\mathbb{K}$ consist of two Lie algebras $\big(\fh,[\cdot,\cdot]_\fh\big)$ and $\big(\fg,[\cdot,\cdot]\big)$ over $\mathbb{K}$ and the tuple of maps,
\begin{equation*}
    \mu_1: \fh\rightarrow\fg,\qquad \mu_2: \fg\wedge\fh \rightarrow \fh\, ,\footnotemark
\end{equation*}
\footnotetext{Here $\wedge$ denotes the skew-symmetric tensor product and $\odot$ denotes the symmetric tensor product.} subject to the following conditions for each $\sfx, \sfx',\sfx'' \in\fg$ and $\sfy,\sfy'\in\fh$:
\begin{enumerate}
    \item The $\fg$-equivariance of $\mu_1$ and the Peiffer identity,
    \begin{equation}
        \mu_1(\mu_2(\sfx,\sfy))=[\sfx,\mu_1(\sfy)]\,,\qquad \mu_2(\mu_1(\sfy),\sfy') =  [\sfy,\sfy']_\fh =- \mu_2(\mu_1(\sfy'),\sfy)\,.\label{pfeif1}
\end{equation}
\noindent Note $\operatorname{ker}\mu_1\subset \fh$ is an Abelian ideal due to the Peiffer identity \eqref{pfeif1}.

\item Graded Jacobi identities,
\begin{align}
\label{ec:gradjacobid}
        0&=[\sfx,[\sfx',\sfx'']]+[\sfx',[\sfx'',\sfx]]+[\sfx'',[\sfx,\sfx']], \\
        0&= \mu_2(\sfx,\mu_2(\sfx',\sfy)) - \mu_2(\sfx',\mu_2(\sfx,\sfy)) - \mu_2([\sfx,\sfx'],\sfy)\,.
    \end{align}
\end{enumerate}
Moreover, we call $\fG$ {\bf balanced} \cite{Zucchini:2021bnn} if it is equipped with a graded symmetric non-degenerate bilinear form $\langle \cdot,\cdot\rangle: \fg\odot\fh \rightarrow \mathbb{K}$ which is invariant
\begin{equation}
    \langle \sfx,\mu_2(\sfx',\sfy)\rangle = \langle [\sfx,\sfx'],\sfy\rangle\,,\qquad \langle \mu_1(\sfy),\sfy'\rangle = \langle \mu_1(\sfy'),\sfy\rangle\label{inv}
\end{equation}
for each $\sfx,\sfx'\in\fg$ and $\sfy,\sfy'\in \fh$.
\end{dfn}

There is a one-to-one correspondence between connected, simply-connected Lie 2-groups and Lie 2-algebras, given by integrating both components $\fh,\fg$ of the Lie 2-algebra $\fG$. Under this correspondence, the Lie 2-algebra differential $\mu_1$ is the derivative of the Lie group map $\bar\mu_1$.


\section{Conventions}

\label{sec:cohomologyconventions}

Throughout this article, we work extensively with various cohomology groups. To keep our conventions transparent and accessible, we summarize them here so that the reader can refer back as needed to verify definitions and computations.

We focus primarily on two complexes, the first of which is the Dolbeault complex on \( \mathbb{C} \setminus \{0\} \), organized as follows:
\begin{equation}
    \begin{tikzcd}
  & \Omega^{(0,0)}(\C\setminus \{0\})  \arrow[r, "\partial"] \arrow[d, "\bar \partial"] & \Omega^{(1,0)}(\C\setminus \{0\}) \arrow[r] \arrow[d, "\bar \partial"] &\cdots \\
   & \Omega^{(0,1)}(\C\setminus \{0\}) \arrow[r, "\partial"] & \Omega^{(1,1)}(\C\setminus \{0\})  \arrow[r] & \cdots
\end{tikzcd}
\end{equation}

In this context, we define the \( \partial \)-cohomology groups by
\begin{equation}
    H_\partial^{(p,q)}(\mathbb{C} \setminus \{0\}) 
    = \frac{ \{ \text{\( \partial \)-closed } (p,q) \text{-forms} \} }
           { \{ \text{\( \partial \)-exact } (p,q) \text{-forms} \} }\,,
\end{equation}
and analogously for the $\bar \partial$-cohomology, by replacing $\partial$ with $\bar \partial$ in the definition above.

The second complex we will frequently use is the \emph{raviolo complex}. Given a three-manifold $M$ equipped with a transverse holomorphic foliation, the raviolo complex is defined as follows:
\begin{equation}
     \begin{tikzcd}
  & \cA^{0,0}(M)  \arrow[r, "\dr'"] \arrow[d, " \partial"] & \cA^{1,0}(M) \arrow[r] \arrow[d, "\partial"] &\cdots \\
   & \cA^{0,1}(M) \arrow[r, "\dr'"] & \cA^{1,1}(M)  \arrow[r] & \cdots
\end{tikzcd}
\end{equation}
In this context, we define the \( \dr' \)-cohomology groups by
\begin{equation}
    H_{\dr'}^{(p,q)}(M) 
    = \frac{ \{ \text{\( \dr' \)-closed } (p,q) \text{-forms} \} }
           { \{ \text{\( \dr' \)-exact } (p,q) \text{-forms} \} }\,\,,
\end{equation}
and analogously for the $ \partial$-cohomology, by replacing $\dr'$ with $ \partial$ in the definition above.

The same notational conventions apply to all other complexes appearing throughout the text.

\section{Cohomology of Lie 2-algebras}

\label{sec:appendixproof}

In the present case, the currents and transformation parameters are valued in the cohomology \( H^{\bullet}(\mathfrak{G}) = V \oplus \mathfrak{n} \), where, recall,
\begin{equation}
    V = \ker(\mu_1) \subset \mathfrak{h}\,, \qquad \mathfrak{n} = \mathfrak{g} / \operatorname{im}(\mu_1) \,.
\end{equation}
To write expressions like \( \langle \tilde{\alpha}, \tilde{H} \rangle \) or \( [\tilde{\alpha}, \tilde{\alpha}'] \), as in the two-dimensional case, it is necessary to verify that these operations descend to cohomology. We include a proof of this fact for completeness.

We begin with the natural projection and inclusion maps
\begin{equation*}
    p = (p_{-1}, p_0): \mathfrak{G} \to H^\bullet(\mathfrak{G}), \qquad 
    \iota = (\iota_{-1}, \iota_0): H^\bullet(\mathfrak{G}) \to \mathfrak{G},
\end{equation*}
which satisfy $p\circ\iota=\id_{H^\bullet(\fG)}$. The equation $\iota\circ p \simeq \id_{\fG}$, on the other hand, only holds \textit{up to homotopy}~\cite{Stasheff:1963} (see also~\cite{Alfonsi_2023}): there exists a linear map \( h: \mathfrak{g} \to \mathfrak{h} \), called the \emph{chain homotopy}, which satisfies
\begin{equation}
    h \circ \mu_1 = \iota_{-1} \circ p_{-1} - \mathrm{id}_{\mathfrak{h}}\,, \qquad 
    \mu_1 \circ h = \iota_0 \circ p_0 - \mathrm{id}_{\mathfrak{g}}\,.
\end{equation}
The existence of such a chain homotopy gives rise to higher-order bracket operations on the cohomology ---  this is the essence of the \textbf{homotopy transfer theorem} \cite{Baez:2003fs,stasheff2018,Arvanitakis:2020rrk}.

\begin{prop} 
\label{prop:nisanalgebra}
The following properties hold
    \begin{enumerate}
        \item The induced graded brackets on $V\oplus \fn$ satisfy the graded Jacobi identity.
        \item The non-degenerate bilinear form $\langle \cdot,\cdot\rangle : \fg \times \fh \to \C$ induces a non-degenerate bilinear form $\langle\cdot,\cdot\rangle:\fn \times V \to \C$.
    \end{enumerate}
Note that the first property guarantees that $\fn$ inherits a Lie bracket from $\fg$, and that $\mu_2$ induces a well-defined derivation on $V$.
\end{prop}

\begin{proof}

To prove (1) we can use this chain homotopy to show that the canonically induced Lie bracket on $\fn$ and the induced action on $V$ 
\begin{equation}
    [\cdot,\cdot]'=p_0\circ [\cdot,\cdot]\circ(\iota_0\otimes\iota_0) \qquad \mu_2' = p_{-1}\circ \mu_2\circ (\iota_0\otimes\iota_{-1})
\end{equation}
satisfy the graded Jacobi identity.

Indeed, to see this, we compute for each $\mathsf{x},\mathsf{x}'\in\fn$ and $\mathsf{y}\in V$ that
\begin{equation*}
\begin{split}
    \mu_2'(\mathsf{x},\mu_2'(\mathsf{x}',\mathsf{y})) &= p_{-1}\Big(\mu_2\big(\iota_0\mathsf{x},(\iota_{-1}p_{-1})\mu_2(\iota_0\mathsf{x}',\iota_{-1}\mathsf{y})\big)\Big)\\
    &=p_{-1}\Big(\mu_2\big(\iota_0\mathsf{x},h\mu_1\mu_2(\iota_0\mathsf{x}',\iota_{-1}\mathsf{y})\big)\Big) + p_{-1}\Big(\mu_2\big(\iota_0\mathsf{x},\mu_2(\iota_0\mathsf{x}',\iota_{-1}\mathsf{y})\big)\Big)\\
    &=p_{-1}\Big(\mu_2\big(\iota_0\mathsf{x},h[\iota_0\mathsf{x}',\mu_1\iota_{-1}\mathsf{y}]\big)\Big) + p_{-1}\Big(\mu_2\big(\iota_0\mathsf{x},\mu_2(\iota_0\mathsf{x}',\iota_{-1}\mathsf{y})\big)\Big)\\
    &= 0 + p_{-1}\Big(\mu_2\big(\iota_0\mathsf{x},\mu_2(\iota_0\mathsf{x}',\iota_{-1}\mathsf{y})\big)\Big) \,,
\end{split}
\end{equation*}
where we have used the identity $\mu_1\mu_2(\cdot,\cdot) = [\cdot,\mu_1 \cdot]$, and the first term vanishes in the last line due to the definition of $\iota_{-1}: V=\operatorname{ker}(\mu_1)\hookrightarrow \fh$.

Similarly, we compute for $\mathsf{x}_1,\mathsf{x}_2,\mathsf{x}_3\in\fn$,
\begin{equation*}
\begin{split}
    [\mathsf{x}_1,[\mathsf{x}_2,\mathsf{x}_3]']' &= p_0\Big([\iota_0\mathsf{x}_1,\iota_0p_0\big([\iota_0\mathsf{x}_2,\iota_0\mathsf{x}_3]\big)]\Big) \\ 
    &= p_0\Big([\iota_0\mathsf{x}_1,\mu_1h\big([\iota_0\mathsf{x}_2,\iota_0\mathsf{x}_3]\big)]\Big) - p_0\Big([\iota_0\mathsf{x}_1,\big([\iota_0\mathsf{x}_2,\iota_0\mathsf{x}_3]\big)]\Big) \\
    &= p_0\Big(\mu_1\big(\mu_2(\iota_0\mathsf{x}_1,h[\iota_0\mathsf{x}_2,\iota_0\mathsf{x}_3])\big)\Big) - p_0\Big([\iota_0\mathsf{x}_1,\big([\iota_0\mathsf{x}_2,\iota_0\mathsf{x}_3]\big)]\Big) \\ 
    &= 0 - p_0\Big([\iota_0\mathsf{x}_1,\big([\iota_0\mathsf{x}_2,\iota_0\mathsf{x}_3]\big)]\Big)\,,
\end{split}
\end{equation*}
where the first term vanishes in the last line due to the definition $p_0: \fg\rightarrow \fn=\fg/\operatorname{im}(\mu_1)$. By summing over cyclic permutations, the remaining quantities in both of the above equations vanish by the graded Jacobi identity satisfied by $\fG$. Thus, property (1) holds.

To prove (2), we must show that the induced bilinear form \(\langle \cdot,\cdot \rangle\) is both well defined and non-degenerate. First, we verify that it is independent of the choice of representatives. Let \(\mathsf{x}, \mathsf{x}' \in \mathfrak{g}\) be representatives of the same class in \(\mathfrak{n} = \mathfrak{g} / \mathrm{im}(\mu_1)\), so that \(\mathsf{x}' = \mathsf{x} + \mu_1(\theta)\) for some \(\theta \in \mathfrak{h}\). Then, for any \(\mathsf{y} \in V = \ker(\mu_1)\), we have
\begin{equation}
    \langle \mathsf{x}', \mathsf{y} \rangle 
    = \langle \mathsf{x} + \mu_1(\theta), \mathsf{y} \rangle 
    = \langle \mathsf{x}, \mathsf{y} \rangle + \langle \mu_1(\theta), \mathsf{y} \rangle\,.
\end{equation}
The second term vanishes due to the compatibility of the bilinear form with \(\mu_1\), namely,
\begin{equation}
    \langle \mu_1(\theta), \mathsf{y} \rangle = \langle \theta, \mu_1(\mathsf{y}) \rangle = 0\,,
\end{equation}
since \(\mathsf{y} \in \ker(\mu_1)\). Therefore, the pairing is well defined on cohomology classes.

To show that the induced form is non-degenerate, we apply the rank-nullity theorem. Since \(\mu_1 : \mathfrak{h} \to \mathfrak{g}\) is linear, we have:
\begin{align}
    \dim \ker(\mu_1) + \dim \mathrm{im}(\mu_1) &= \dim \mathfrak{h}\,, \\
    \dim \mathrm{coker}(\mu_1) + \dim \mathrm{im}(\mu_1) &= \dim \mathfrak{g}\,.
\end{align}
Since the original bilinear form is invertible we have that $\dim \fh =\dim \fg$ so that subtracting these equations gives
\begin{equation}
   \dim V = \dim \ker(\mu_1) = \dim \mathrm{coker}(\mu_1)=\dim \mathfrak{n}\,,
\end{equation}
which implies that the restriction of $\langle \cdot,\cdot\rangle$ to $V\oplus \fn$ is non-degenerate.

\end{proof}

\printbibliography

@article{Oh:2019mcg,
    author = "Oh, Jihwan and Yagi, Junya",
    title = "{Poisson vertex algebras in supersymmetric field theories}",
    eprint = "1908.05791",
    archivePrefix = "arXiv",
    primaryClass = "hep-th",
    doi = "10.1007/s11005-020-01290-0",
    journal = "Lett. Math. Phys.",
    volume = "110",
    number = "8",
    pages = "2245--2275",
    year = "2020"
}

@article{Beem:2018fng,
    author = "Beem, Christopher and Ben-Zvi, David and Bullimore, Mathew and Dimofte, Tudor and Neitzke, Andrew",
    title = "{Secondary products in supersymmetric field theory}",
    eprint = "1809.00009",
    archivePrefix = "arXiv",
    primaryClass = "hep-th",
    doi = "10.1007/s00023-020-00888-3",
    journal = "Annales Henri Poincare",
    volume = "21",
    number = "4",
    pages = "1235--1310",
    year = "2020"
}

@article{Budzik:2025zvu,
    author = "Budzik, Kasia and Kulp, Justin",
    title = "{Loop Corrected Supercharges from Holomorphic Anomalies}",
    eprint = "2512.07771",
    archivePrefix = "arXiv",
    primaryClass = "hep-th",
    month = "12",
    year = "2025"
}

@article{Zeng:2023qqp,
    author = "Zeng, Keyou",
    title = "{Twisted Holography and Celestial Holography from Boundary Chiral Algebra}",
    eprint = "2302.06693",
    archivePrefix = "arXiv",
    primaryClass = "hep-th",
    doi = "10.1007/s00220-023-04917-0",
    journal = "Commun. Math. Phys.",
    volume = "405",
    number = "1",
    pages = "19",
    year = "2024"
}

@article{Bunk:2020rju,
    author = {Bunk, Severin and M{\"u}ller, Lukas and Szabo, Richard J.},
    title = "{Smooth 2-Group Extensions and Symmetries of Bundle Gerbes}",
    eprint = "2004.13395",
    archivePrefix = "arXiv",
    primaryClass = "math.DG",
    reportNumber = "Hamburger Beitrage zur Mathematik Nr. 834, ZMP-HH/20-10, EMPG-20-08",
    doi = "10.1007/s00220-021-04099-7",
    journal = "Commun. Math. Phys.",
    volume = "384",
    number = "3",
    pages = "1829--1911",
    year = "2021"
}

@article{Schommer_Pries_2011,
   title={Central extensions of smooth 2–groups and a finite-dimensional string 2–group},
   volume={15},
   ISSN={1465-3060},
   url={http://dx.doi.org/10.2140/gt.2011.15.609},
   DOI={10.2140/gt.2011.15.609},
   number={2},
   journal={Geometry \& Topology},
   publisher={Mathematical Sciences Publishers},
   author={Schommer-Pries, Christopher-J},
   year={2011},
   month=may, pages={609–676} }

@article{Garner:2023zko,
    author = "Garner, Niklas and Raghavendran, Surya and Williams, Brian R.",
    title = "{Higgs and Coulomb branches from superconformal raviolo vertex algebras}",
    eprint = "2310.08524",
    archivePrefix = "arXiv",
    primaryClass = "math.QA",
    doi = "10.1016/j.aim.2025.110566",
    journal = "Adv. Math.",
    volume = "482",
    pages = "110566",
    year = "2025"
}

@ARTICLE{Ludewig:2023,
  title    = "Lie 2-groups from loop group extensions",
  author   = "Ludewig, Matthias and Waldorf, Konrad",
  journal  = "Journal of Homotopy and Related Structures",
  volume   =  19,
  number   =  4,
  pages    = "597--633",
  month    =  dec,
  year     =  2024
}

@article{Kemp_2025,
   title={Infinitesimal 2-braidings from 2-shifted Poisson structures},
   volume={212},
   ISSN={0393-0440},
   url={http://dx.doi.org/10.1016/j.geomphys.2025.105456},
   DOI={10.1016/j.geomphys.2025.105456},
   journal={Journal of Geometry and Physics},
   publisher={Elsevier BV},
   author={Kemp, Cameron and Laugwitz, Robert and Schenkel, Alexander},
   year={2025},
   month=jun, pages={105456} }

@article{Dimofte:2025oqf,
    author = "Dimofte, Tudor and Niu, Wenjun and Py, Victor",
    title = "{Line Operators in 3d Holomorphic QFT: Meromorphic Tensor Categories and dg-Shifted Yangians}",
    eprint = "2508.11749",
    archivePrefix = "arXiv",
    primaryClass = "hep-th",
    month = "8",
    year = "2025"
}

@article{Cirio:2012be,
    author = "Cirio, Lucio Simone and Martins, Joao Faria",
    title = "{Categorifying the $sl(2,C)$ Knizhnik-Zamolodchikov Connection via an Infinitesimal 2-Yang-Baxter Operator in the String Lie-2-Algebra}",
    eprint = "1207.1132",
    archivePrefix = "arXiv",
    primaryClass = "hep-th",
journal = "Adv. Theor. Math. Phys.",
volume="21", 
number="1", 
paegs="147 - 229", 
year = "2017"
}

@ARTICLE{Alekseev1990-hq,
  title    = "Quantum groups and {WZNW} models",
  author   = "Alekseev, A and Shatashvili, S",
  journal  = "Communications in Mathematical Physics",
  volume   =  133,
  number   =  2,
  pages    = "353--368",
  month    =  oct,
  year     =  1990
}

@book{frenkel2004vertex,
  title={Vertex Algebras and Algebraic Curves},
  author={Frenkel, E. and Ben-Zvi, D.},
  isbn={9780821836743},
  lccn={04051904},
  series={Mathematical surveys and monographs},
  url={https://books.google.com.ar/books?id=r82mAgAAQBAJ},
  year={2004},
  publisher={American Mathematical Society}
}

@article{Frenkel:1994em,
    author = "Frenkel, Edward and Kac, Victor and Radul, Andrey and Wang, Wei-Qiang",
    title = "{W(1+infinity) and W(gl(N)) with central charge N}",
    eprint = "hep-th/9405121",
    archivePrefix = "arXiv",
    doi = "10.1007/BF02108332",
    journal = "Commun. Math. Phys.",
    volume = "170",
    pages = "337--358",
    year = "1995"
}

@book{Kac1998,
  author    = {Kac, Victor G.},
  title     = {Vertex Algebras for Beginners},
  series    = {University Lecture Series},
  volume    = {10},
  publisher = {American Mathematical Society},
  year      = {1998},
  address   = {Providence, RI},
  isbn      = {082181396X},
  pages     = {201}
}

@article{Budzik:2022mpd,
    author = "Budzik, Kasia and Gaiotto, Davide and Kulp, Justin and Wu, Jingxiang and Yu, Matthew",
    title = "{Feynman diagrams in four-dimensional holomorphic theories and the Operatope}",
    eprint = "2207.14321",
    archivePrefix = "arXiv",
    primaryClass = "hep-th",
    doi = "10.1007/JHEP07(2023)127",
    journal = "JHEP",
    volume = "07",
    pages = "127",
    year = "2023"
}

@article{Gaiotto:2024gii,
    author = "Gaiotto, Davide and Kulp, Justin and Wu, Jingxiang",
    title = "{Higher operations in perturbation theory}",
    eprint = "2403.13049",
    archivePrefix = "arXiv",
    primaryClass = "hep-th",
    doi = "10.1007/JHEP05(2025)230",
    journal = "JHEP",
    volume = "05",
    pages = "230",
    year = "2025"
}

@article{Felder:2025bsz,
    author = "Felder, Laura O. and Gui, Zhengping and Young, Charles A. S.",
    title = "{Higher Chiral Algebras in a Polysimplicial Model}",
    eprint = "2506.09728",
    archivePrefix = "arXiv",
    primaryClass = "math.QA",
    month = "6",
    year = "2025"
}

@article{Hollands:2023txn,
    author = "Hollands, Stefan and Wald, Robert M.",
    title = "{The Operator Product Expansion in Quantum Field Theory}",
    eprint = "2312.01096",
    archivePrefix = "arXiv",
    primaryClass = "hep-th",
    month = "12",
    year = "2023"
}

@article{Wang:2024tjf,
    author = "Wang, Minghao and Williams, Brian R.",
    title = "{Factorization algebras from topological-holomorphic field theories}",
    eprint = "2407.08667",
    archivePrefix = "arXiv",
    primaryClass = "math-ph",
    month = "7",
    year = "2024"
}

@article{Garner:2023wrc,
    author = "Garner, Niklas and Raghavendran, Surya and Williams, Brian R.",
    title = "{Enhanced symmetries in minimally-twisted three-dimensional supersymmetric theories}",
    eprint = "2310.08516",
    archivePrefix = "arXiv",
    primaryClass = "hep-th",
    month = "10",
    year = "2023"
}

@article{FAONTE2019389,
title = {Higher Kac–Moody algebras and moduli spaces of G-bundles},
journal = {Advances in Mathematics},
volume = {346},
pages = {389-466},
year = {2019},
issn = {0001-8708},
doi = {https://doi.org/10.1016/j.aim.2019.01.040},
url = {https://www.sciencedirect.com/science/article/pii/S0001870819300763},
author = {Giovanni Faonte and Benjamin Hennion and Mikhail Kapranov},
}

@article{Duchamp:1979,
 author = {Duchamp, T. and Kalka, M.},
 title = {Deformation theory for holomorphic foliations},
 fjournal = {Journal of Differential Geometry},
 journal = {J. Differ. Geom.},
 issn = {0022-040X},
 volume = {14},
 pages = {317--337},
 year = {1979},
 language = {English},
 doi = {10.4310/jdg/1214435099},
 zbMATH = {3705536},
 Zbl = {0451.57015}
}

@article{Alfonsi:2024qdr,
    author = "Alfonsi, Luigi and Kim, Hyungrok and Young, Charles A. S.",
    title = "{Raviolo vertex algebras, cochains and conformal blocks}",
    eprint = "2401.11917",
    archivePrefix = "arXiv",
    primaryClass = "math.QA",
    month = "1",
    year = "2024"
}

@article{stasheff2018,
  title={$L_\infty$ and $A_\infty$ structures: then and now},
  author={Stasheff, Jim},
  journal={Higher Structures},
volume = {3},
number ={1},
pages = {292–326},
  year={2019}
}

@article{Arvanitakis:2020rrk,
    author = "Arvanitakis, Alex S. and Hohm, Olaf and Hull, Chris and Lekeu, Victor",
    title = "{Homotopy Transfer and Effective Field Theory I: Tree-level}",
    eprint = "2007.07942",
    archivePrefix = "arXiv",
    primaryClass = "hep-th",
    reportNumber = "Imperial-TP-2020-CH-02",
    doi = "10.1002/prop.202200003",
    journal = "Fortsch. Phys.",
    volume = "70",
    number = "2-3",
    pages = "2200003",
    year = "2022"
}

@article{Aganagic:2017tvx,
    author = "Aganagic, Mina and Costello, Kevin and McNamara, Jacob and Vafa, Cumrun",
    title = "{Topological Chern-Simons/Matter Theories}",
    eprint = "1706.09977",
    archivePrefix = "arXiv",
    primaryClass = "hep-th",
    month = "6",
    year = "2017"
}

@article{Zucchini:2021bnn,
    author = "Zucchini, Roberto",
    title = "{4-d Chern-Simons Theory: Higher Gauge Symmetry and Holographic Aspects}",
    eprint = "2101.10646",
    archivePrefix = "arXiv",
    primaryClass = "hep-th",
    reportNumber = "DIFA UNIBO 2021",
    doi = "10.1007/JHEP06(2021)025",
    journal = "JHEP",
    volume = "06",
    pages = "025",
    year = "2021"
}

@article{Alfonsi_2023,
	doi = {10.1016/j.geomphys.2023.104903},
  
	url = {https://doi.org/10.1016%2Fj.geomphys.2023.104903},
  
	year = 2023,
	month = {9},
  
	publisher = {Elsevier {BV}
},
  
	volume = {191},
  
	pages = {104903},
  
	author = {Luigi Alfonsi and Charles Young},
  
	title = {Higher current algebras, homotopy Manin triples, and a rectilinear adelic complex},
  
	journal = {Journal of Geometry and Physics}
}

@article{Stasheff:1963,
 ISSN = {00029947},
 URL = {http://www.jstor.org/stable/1993608},
 author = {James Dillon Stasheff},
 journal = {Transactions of the American Mathematical Society},
 number = {2},
 pages = {275--292},
 publisher = {American Mathematical Society},
 title = {Homotopy Associativity of H-Spaces. I},
 urldate = {2023-07-29},
 volume = {108},
 year = {1963}
}

@book{costello_gwilliam_2016, place={Cambridge}, series={New Mathematical Monographs}, title={Factorization Algebras in Quantum Field Theory}, volume={1}, DOI={10.1017/9781316678626}, publisher={Cambridge University Press}, author={Costello, Kevin and Gwilliam, Owen}, year={2016}, collection={New Mathematical Monographs}}

@article{Porst2008Strict2A,
  title={Strict 2-Groups are Crossed Modules},
  author={Sven-S. Porst},
  journal={arXiv: Category Theory},
  year={2008}
}

@article{Garner:2023zqn,
    author = "Garner, Niklas and Williams, Brian R.",
    title = "{Raviolo vertex algebras}",
    eprint = "2308.04414",
    archivePrefix = "arXiv",
    primaryClass = "math.QA",
    month = "8",
    year = "2023"
}

@article{Gwilliam:2021zkv,
    author = "Gwilliam, Owen and Rabinovich, Eugene and Williams, Brian R.",
    title = "{Quantization of topological-holomorphic field theories: local aspects}",
    eprint = "2107.06734",
    archivePrefix = "arXiv",
    primaryClass = "math-ph",
    month = "7",
    year = "2021"
}

@article{Chen:2023integrable,
    author = "Chen, Hank and Girelli, Florian",
    title = "{Integrability from categorification}",
    eprint = "2307.03831",
    archivePrefix = "arXiv",
    primaryClass = "math-ph",
    month = "7",
    year = "2023"
}

@article{Chen:2023tjf,
    author = "Chen, Hank and Girelli, Florian",
    title = "{Categorified Quantum Groups and Braided Monoidal 2-Categories}",
    eprint = "2304.07398",
    archivePrefix = "arXiv",
    primaryClass = "math.QA",
    month = "4",
    year = "2023"
}

@book{Wagemann+2021,
author = {Friedrich Wagemann},
doi = {doi:10.1515/9783110750959},
url = {https://doi.org/10.1515/9783110750959},
title = {Crossed Modules},
year = {2021},
publisher = {De Gruyter},
ISBN = {9783110750959},
lastchecked = {2022-06-28}
}

@article{Baez:2004,
author = {Baez, John C. and Lauda, Aaron D.},
journal = {Theory and Applications of Categories [electronic only]},
language = {eng},
pages = {423-491},
publisher = {Mount Allison University, Department of Mathematics and Computer Science, Sackville},
title = {Higher-dimensional algebra. V: 2-Groups.},
url = {http://eudml.org/doc/124217},
volume = {12},
year = {2004},
}

@article{chen:2022,
  title = {Categorified Drinfel'd double and $BF$ theory: 2-groups in 4D},
  author = {Chen, Hank and Girelli, Florian},
  journal = {Phys. Rev. D},
  volume = {106},
  issue = {10},
  pages = {105017},
  numpages = {35},
  year = {2022},
  month = {11},
  publisher = {American Physical Society},
  doi = {10.1103/PhysRevD.106.105017},
  url = {https://link.aps.org/doi/10.1103/PhysRevD.106.105017}
}

@article{Crane:1994ty,
    author = "Crane, Louis and Frenkel, Igor",
    title = "{Four-dimensional topological field theory, Hopf categories, and the canonical bases}",
    eprint = "hep-th/9405183",
    archivePrefix = "arXiv",
    doi = "10.1063/1.530746",
    journal = "J. Math. Phys.",
    volume = "35",
    pages = "5136--5154",
    year = "1994"
}

@article{Baez:1995xq,
    author = "Baez, J. C. and Dolan, J.",
    title = "{Higher dimensional algebra and topological quantum field theory}",
    eprint = "q-alg/9503002",
    archivePrefix = "arXiv",
    doi = "10.1063/1.531236",
    journal = "J. Math. Phys.",
    volume = "36",
    pages = "6073--6105",
    year = "1995"
}

@article{Baez:2003fs,
	archiveprefix = {arXiv},
	author = {Baez, John C. and Crans, Alissa S.},
	eprint = {math/0307263},
	journal = {Theor. Appl. Categor.},
	pages = {492--528},
	title = {{Higher-Dimensional Algebra VI: Lie 2-Algebras}},
	volume = {12},
	year = {2004}}

@article{Chen:2024axr,
    author = "Chen, Hank and Liniado, Joaquin",
    title = "{Higher gauge theory and integrability}",
    eprint = "2405.18625",
    archivePrefix = "arXiv",
    primaryClass = "hep-th",
    doi = "10.1103/PhysRevD.110.086017",
    journal = "Phys. Rev. D",
    volume = "110",
    number = "8",
    pages = "086017",
    year = "2024"
}

@article{Kim:2019owc,
	archiveprefix = {arXiv},
	author = {Kim, Hyungrok and Saemann, Christian},
	doi = {10.1088/1751-8121/ab8ef2},
	eprint = {1911.06390},
	journal = {J. Phys. A},
	number = {44},
	pages = {445206},
	primaryclass = {hep-th},
	reportnumber = {EMPG-19-24},
	title = {{Adjusted parallel transport for higher gauge theories}},
	volume = {53},
	year = {2020}}

@article{Chen:2013,
	author = {Chen, Zhuo and Sti{\'e}non, Mathieu and Xu, Ping},
	doi = {10.1016/j.geomphys.2013.01.006},
	issn = {0393-0440},
	journal = {Journal of Geometry and Physics},
	month = {Jun},
	pages = {59--68},
	publisher = {Elsevier BV},
	title = {Weak Lie 2-bialgebras},
	url = {http://dx.doi.org/10.1016/j.geomphys.2013.01.006},
	volume = {68},
	year = {2013}}

@article{Chen:2012gz,
	archiveprefix = {arXiv},
	author = {Chen, Zhuo and Sti{\'e}non, Mathieu and Xu, Ping},
	doi = {10.4310/jdg/1367438648},
	eprint = {1202.0079},
	journal = {J. Diff. Geom.},
	number = {2},
	pages = {209--240},
	primaryclass = {math.DG},
	title = {{Poisson 2-groups}},
	volume = {94},
	year = {2013}}

@article{Bai_2013,
	author = {Chengming Bai and Yunhe Sheng and Chenchang Zhu},
	doi = {10.1007/s00220-013-1712-3},
	journal = {Communications in Mathematical Physics},
	month = {4},
	number = {1},
	pages = {149--172},
	publisher = {Springer Science and Business Media {LLC}},
	title = {Lie 2-Bialgebras},
	url = {https://doi.org/10.1007%2Fs00220-013-1712-3},
	volume = {320},
	year = 2013}

@article{Costello:2020ndc,
    author = "Costello, Kevin and Dimofte, Tudor and Gaiotto, Davide",
    title = "{Boundary Chiral Algebras and Holomorphic Twists}",
    eprint = "2005.00083",
    archivePrefix = "arXiv",
    primaryClass = "hep-th",
    doi = "10.1007/s00220-022-04599-0",
    journal = "Commun. Math. Phys.",
    volume = "399",
    number = "2",
    pages = "1203--1290",
    year = "2023"
}

@article{Khan:2025rah,
    author = "Khan, Ahsan Z. and Zeng, Keyou",
    title = "{Poisson Vertex Algebras and Three-Dimensional Gauge Theory}",
    eprint = "2502.13227",
    archivePrefix = "arXiv",
    primaryClass = "hep-th",
    month = "2",
    year = "2025"
}

@article{Griffin:2025add,
    author = "Griffin, Colton",
    title = "{Cohomological vertex algebras}",
    eprint = "2501.18720",
    archivePrefix = "arXiv",
    primaryClass = "math.QA",
    month = "1",
    year = "2025"
}

@article{Rawnsley,
author = {Rawnsley, John},
year = {1979},
month = {03},
pages = {391-391},
title = {Flat Partial Connections and Holomorphic Structures in $C^\infty$ Vector Bundles},
volume = {73},
journal = {Proceedings of The American Mathematical Society - PROC AMER MATH SOC},
doi = {10.1090/S0002-9939-1979-0518527-X}
}

@article{Garner:2023izn,
    author = "Garner, Niklas and Paquette, Natalie M.",
    title = "{Twistorial monopoles \& chiral algebras}",
    eprint = "2305.00049",
    archivePrefix = "arXiv",
    primaryClass = "hep-th",
    doi = "10.1007/JHEP08(2023)088",
    journal = "JHEP",
    volume = "08",
    pages = "088",
    year = "2023"
}

@article{Zeng:2021zef,
    author = "Zeng, Keyou",
    title = "{Monopole operators and bulk-boundary relation in holomorphic topological theories}",
    eprint = "2111.00955",
    archivePrefix = "arXiv",
    primaryClass = "hep-th",
    doi = "10.21468/SciPostPhys.14.6.153",
    journal = "SciPost Phys.",
    volume = "14",
    number = "6",
    pages = "153",
    year = "2023"
}

@article{Beem:2013sza,
    author = "Beem, Christopher and Lemos, Madalena and Liendo, Pedro and Peelaers, Wolfger and Rastelli, Leonardo and van Rees, Balt C.",
    title = "{Infinite Chiral Symmetry in Four Dimensions}",
    eprint = "1312.5344",
    archivePrefix = "arXiv",
    primaryClass = "hep-th",
    reportNumber = "YITP-SB-13-45, CERN-PH-TH-2013-311, HU-EP-13-78",
    doi = "10.1007/s00220-014-2272-x",
    journal = "Commun. Math. Phys.",
    volume = "336",
    number = "3",
    pages = "1359--1433",
    year = "2015"
}

@article{Belavin:1984vu,
    author = "Belavin, A. A. and Polyakov, Alexander M. and Zamolodchikov, A. B.",
    editor = "Khalatnikov, I. M. and Mineev, V. P.",
    title = "{Infinite Conformal Symmetry in Two-Dimensional Quantum Field Theory}",
    reportNumber = "CERN-TH-3827",
    doi = "10.1016/0550-3213(84)90052-X",
    journal = "Nucl. Phys. B",
    volume = "241",
    pages = "333--380",
    year = "1984"
}

\end{document}